\documentclass[final, journal, oneside, twocolumn, letterpaper, romanappendices]{IEEEtran}

\usepackage{ifpdf}
\usepackage[noadjust]{cite}
\ifCLASSINFOpdf
   \usepackage[pdftex]{graphicx}
\else
\fi
\usepackage[cmex10]{amsmath}
\usepackage{algorithmic}
\usepackage{array}
\usepackage{amsthm}
\usepackage{mdwmath}
\usepackage{mdwtab}
\usepackage{eqparbox}
\usepackage{hyperref}
\usepackage{amssymb,amsmath}
\usepackage{color}
\usepackage{graphicx}
\usepackage{caption}
\usepackage{subcaption}

\newcounter{MYtempeqncnt}

\newtheorem{theorem}{Theorem}

\newtheorem{lemma}[theorem]{Lemma}

\newtheorem{definition}{Definition}

\newcommand{\field}[1]{\mathbb{#1}}
\newcommand{\R}{\field{R}} 

\newcommand{\cmle}{\text{\rm cml}}

\newcommand{\minimax}{\text{\rm minimax}}

\newcommand{\bayesQ}{\hat{X}_Q^{\text{\rm Bayes}}}
\newcommand{\bayesQstar}{\hat{X}_{Q^*}^{\text{\rm Bayes}}}
\newcommand{\bayesPtheta}{\hat{X}_{P_{\theta}}^{\text{\rm Bayes}}}

\newcommand{\cmlbayesQ}{\cmle(\theta,\bayesQ)}
\newcommand{\cmlbayesQstar}{\cmle(\theta,\bayesQstar)}
\newcommand{\cmlbayesPtheta}{\cmle(\theta,\bayesPtheta)}

\newcommand{\cmlxhat}{\cmle(\theta,\hat{X})}
\newcommand{\regretxhat}{R(\theta,\hat{X})}
\newcommand{\regretbayesQ}{R(\theta,\bayesQ)}
\newcommand{\E}{\field{E}}
\newcommand{\mathP}{\mathcal{P}}
\newcommand{\mathQ}{\mathcal{Q}}
\newcommand{\mathX}{\mathcal{X}}
\newcommand{\mathY}{\mathcal{Y}}
\newcommand{\mathT}{\mathcal{T}}
\newcommand{\mathN}{\mathcal{N}}

\newcommand{\mathS}{\mathcal{S}}

\newcommand{\diag}{\mbox{diag}}
\newcommand{\conv}{\text{\rm conv}}
\newcommand{\cl}{\text{\rm cl}}

\newcommand{\bea}{\begin{eqnarray}}
\newcommand{\eea}{\end{eqnarray}}
\newcommand{\beas}{\begin{eqnarray*}}
\newcommand{\eeas}{\end{eqnarray*}}
\usepackage{amsfonts}

\newcommand{\suchthat}{\;\ifnum\currentgrouptype=16 \middle\fi|\;}
\allowdisplaybreaks[1]

\begin{document}
\title{Minimax Filtering Regret via Relations Between Information and Estimation}

\author{Albert No, \emph{Student Member, IEEE}, and Tsachy Weissman, \emph{Fellow, IEEE}}
\maketitle
\begin{abstract}
We investigate the problem of continuous-time causal estimation under a minimax criterion. Let $X^T = \{X_t,0\leq t\leq T\}$ be governed by the probability law $P_{\theta}$ from a class of possible laws indexed by $\theta \in \Lambda$, and $Y^T$ be the noise corrupted observations of $X^T$ available to the estimator. We characterize the estimator minimizing the worst case regret, where regret is the difference between the causal estimation loss of the estimator and that of the optimum estimator.

\let\thefootnote\relax\footnotetext{This paper was presented in part at the 2013 IEEE International Symposium on Information Theory.This work was supported by the NSF
Center for Science of Information under Grant Agreement CCF-0939370.

The authors are with the Department of Electrical Engineering, Stanford University, Stanford, CA 94305 USA (e-mail: albertno@stanford.edu; tsachy@stanford.edu).

Copyright (c) 2014 IEEE. Personal use of this material is permitted.  However, permission to use this material for any other purposes must be obtained from the IEEE by sending a request to pubs-permissions@ieee.org.}

One of the main contributions of this paper is characterizing the minimax estimator, showing that it is in fact a Bayesian estimator. We then relate minimax regret to the channel capacity when the channel is either Gaussian or Poisson. In this case, we characterize the minimax regret and the minimax estimator more explicitly. If we further assume that the uncertainty set consists of deterministic signals, the worst case regret is exactly equal to the corresponding channel capacity, namely the maximal mutual information attainable across the channel among all possible distributions on the uncertainty set of signals. The corresponding minimax estimator is the Bayesian estimator assuming the capacity-achieving prior. Using this relation, we also show that the capacity achieving prior coincides with the least favorable input. Moreover, we show that this minimax estimator is not only minimizing the worst case regret but also essentially minimizing regret for ``most" of the other sources in the uncertainty set.

We present a couple of examples for the construction of an minimax filter via an approximation of the associated capacity achieving distribution. 

\end{abstract}
\IEEEpeerreviewmaketitle

\section{Introduction}\label{sec:introduction}
Recent work on relations between information and estimation has shown fundamental links between the causal estimation error and information theoretic quantities. In \cite{duncan70}, Duncan showed that causal estimation error of an additive white Gaussian noise (AWGN) corrupted signal is equal to the mutual information between the input and output processes divided by the signal-to-noise ratio. In \cite{weissman10}, Weissman extended the result to the case of mismatched estimation, where the estimator assumes that the input signal is governed by a law $Q$ while its true law is $P$. In this case, the cost of mismatch, which is half the difference between the mismatched causal estimation error and the optimum (non-mismatched) causal estimation error, is given by the relative entropy between the laws of output processes when the input processes have laws $P$ and $Q$, respectively. In \cite{atar12}, Atar and Weissman showed that parallel information-estimation relations exist in the Poisson channel for both mismatched and non-mismatched settings.

 In this paper, we investigate the continuous-time causal estimation problem. We assume that the input process is governed by a probability law from a known uncertainty class $\mathP$ although the estimator does not know the true law. In particular, suppose that the input process is governed by a law $P_{\theta}\in\mathP$, where $\theta\in\Lambda$ and $\Lambda$ is the uncertainty set known to the decoder. In this setting, it is natural to consider the minimax estimator which minimizes the worst case regret, where regret is defined as the difference between the causal estimation error of the estimator and that of the optimum estimator. If there is a minimum achieving estimator, we will call it a \emph{minimax estimator} or minimax filter. One of the main contributions of this paper is characterizing the minimax estimator, showing that it is in fact a Bayesian estimator under the distribution which is the capacity-achieving mixture of distributions associated with the channel whose input is a source in the uncertainty set.

We can find similar arguments in classical universal source coding theory. In universal source coding theory, the encoder only knows that the source is governed by some law from an uncertainty set. The goal is to construct the universal code that minimizes the gap between its expected code length and that under the optimum encoding strategy for the true law. 
Redundancy capacity theory \cite{gallager76} tells us that the minimum of the worst case redundancy (minimax redundancy) coincides with the maximum mutual information between  input and output of the channel whose input is a choice of a law from the uncertainty set and whose output is a realization of that law. 

Using these ideas, we show similar results for our causal estimation problem. If the channel is either Gaussian or Poisson, we can combine the results of mismatched estimation and the above redundancy capacity theorem in order to relate the minimax regret to the corresponding channel capacity.  Indeed, the minimax regret turns out to equal to the maximum mutual information between the input index and the corresponding output which we shall refer to as \emph{regret capacity}. Moreover, the minimax filter is Bayesian with respect to the same prior that achieves maximum mutual information. Therefore, if we know the distribution that maximizes mutual information, we can induce the minimax estimator. Further, we shall see that if the class of measures $\mathP$ is a set of deterministic signals, this mutual information reduces to the mutual information between input and output processes $X^T$ and $Y^T$. This allows us to harness well known results from channel coding to characterize and construct the minimax filter. 

The relation between the capacity achieving prior and the minimax filter gives us a new link between estimation and information which is the probability law over input signals that results in the worst causal mean loss. In particular, using the regret-capacity theorem, we show that the capacity achieving prior coincides with the least favorable input. 

Since the goal in minimax estimation is to minimize the worst case regret, one may argue that the minimax estimator might not be a good estimator for many other sources in the class. However, in universal source coding theory, Merhav and Feder \cite{merhav95} showed that the minimax encoder works well for ``most" distributions in the uncertainty set, where ``most'' is measured with respect to the capacity-achieving prior which is argued to be the ``right'' prior. Indeed, the framework of \cite{merhav95} strengthened and generalized the results of this nature that were established for parametric uncertainty sets by Rissanen in \cite{rissanen84}. We apply this idea to our minimax estimation setting. These results imply that the minimax estimator not only minimizes the worst case error, but does essentially at least as well as any other estimator for most sources. 

Our results for the Gaussian and the Poisson channel carry over to accommodate the presence of feedback. In this paper, feedback means that the input process at time $t$, $X_t$, is also affected by previous outputs $\{Y_s:0\leq s<t\}$. We show that all the theorems are still valid in this case by substituting mutual information with directed information.

The rest of the paper is organized as follows. Section \ref{sec:problem setting} describes the concrete problem setting. In Section \ref{sec:main result}, we present and discuss the main results. Relation between the capacity achieving prior and the least favorable input is presented in Section \ref{sec:least favorable input}. Section \ref{sec:proof} provides proofs of the theorems. In Sections \ref{sec:examples} and \ref{sec:experiments}, we provide examples of experiments with simulated signals. We conclude with a summary in Section \ref{sec:conclusions}.

\section{Problem Setting}\label{sec:problem setting}
Let the right-continuous input process $X^T = \{X_t,0\leq t\leq T\}$ be governed by the probability law $P_{\theta}$ from some class of possible laws indexed by $\theta \in \Lambda$. Throughout the paper, we will assume that the collection of laws $\mathP = \{P_{\theta}:\theta\in\Lambda\}$ is tight. $\mathP$ and $\Lambda$ are uncertainty sets known to the estimator. Let $Y^T$ be the noise corrupted observations of $X^T$, and therefore the probability law of $Y^T$ also depends on the specific $\theta\in\Lambda$. However, we assume that the noise corruption mechanism $P_{Y^T|X^T}$ is fixed and known to the decoder. Denote the input and reconstruction alphabets by $\mathX$ and $\hat{\mathX}$, respectively. In other words, $X_t\in\mathX$ and $\hat{X}_t\in\hat{\mathX}$, where both $\mathX$ and $\hat{\mathX}$ are closed subsets of $\R$. Let the measurable\footnote{From this point on we tacitly assume measurability of all functions introduced.} $l(\cdot,\cdot):\mathX\times\hat{\mathX}\mapsto [0,\infty)$ be a given loss function. For simplicity and transparency of our arguments, we assume that $l(\cdot,\cdot)$ satisfies the following properties:
\begin{itemize}
\item[(P1)] $l(x,\hat{x})$ is a lower semi-continuous convex function in $\hat{x}$;
\item[(P2)] $\min_{\hat{x}\in\hat{\mathX}} \E[l(X,\hat{x})] = \E[l(X,\E[X])]$ for all random variables $X$ on $\mathX$.
\end{itemize}
The squared error loss function and the \emph{natural loss function} $l(x,\hat{x}) = x\log(\frac{x}{\hat{x}})-x+\hat{x}$, introduced in \cite{atar12}, are examples of loss functions satisfying these properties. Note that all Bregman loss functions satisfy (P2). Moreover, if $\E[X]$ is a unique minimizer of $\E[l(X,\hat{x})]$ for all random variables $X$ (i.e., (P2) with uniqueness), then $l(\cdot,\cdot)$ is a Bregman loss function (up to an additive constant) \cite{banerjee05}. However, Bregman loss functions are not convex in the second argument in general.

Define the causal estimator $\hat{X}_t(\cdot)$ as a function of the output process up to time $t$, i.e., $Y^t = \{Y_s,0\leq s\leq t\}$ and also define the causal mean loss associated with the filter  $\hat{X} = \{ \hat{X}_t( \cdot ), 0 \leq t \leq T \}$ by
\begin{align*}
\cmle(\theta,\hat{X}) = \E_{P_{\theta}}\left[\int_0^T l(X_t,\hat{X}_t(Y^t)) dt\right]
\end{align*}
where $\E_{P_{\theta}}[\cdot]$ denotes expectation under $P_{\theta}\times P_{Y^T|X^T}$.  We will use $\E_{P_{\theta}}[\cdot|Y^t]$ in the rest of paper which denotes conditional expectation under $P_{\theta}\times P_{Y^T|X^T}$.

\section{Main Results}\label{sec:main result}

\subsection{Minimax Causal Estimation Criterion}\label{subsec:minimax causal estimation criterion}
If the estimator knows the true law $P_{\theta}$, property (P2) implies that the optimum filter will be the Bayesian estimator with respect to the law $P_{\theta}$, i.e., the estimate at time $t$ will be $\E_{P_{\theta}}[X_t|Y^t]$. However, since the estimator does not know the true law $P_{\theta}$, the estimator can be optimized for law $Q$ (while the active law remains $P_{\theta}$). Then the estimator is the Bayesian estimator $\bayesQ$, 
where $\bayesQ=\{\E_Q[X_t|\cdot] : 0\leq t\leq T\}$ denotes the collection of Bayesian filter under prior $Q$, i.e., the estimate at time $t$ will be $\bayesQ(Y^t)=\E_Q[X_t|Y^t]$. The corresponding mismatched causal mean loss will be 
\begin{align*}
\cmlbayesQ = \E_{P_{\theta}}\left[\int_0^T l(X_t,\E_Q[X_t|Y^t]) dt\right].
\end{align*}

We can treat $\cmlbayesPtheta$ as our benchmark since it minimizes the causal mean loss when the $P_\theta$ is exactly known. Therefore, we define regret of the filter $\hat{X}$ when the active source is $P_\theta$ by
\begin{align*}
\regretxhat =\cmlxhat -\cmlbayesPtheta.
\end{align*}

Since we do not have a prior on $\theta$, it is natural to seek to minimize the  worst-case regret over all possible $\theta\in\Lambda$. Specifically, define $\minimax(\Lambda)$ as
\beas
\minimax(\Lambda) =\inf_{\hat{X}} \sup_{\theta\in\Lambda} \regretxhat,
\eeas
where the infimum is over all possible filters. If there exists an infimum achieving $\hat{X}$, we will say $\hat{X}$ is the minimax filter.

\subsection{Statement of Results}\label{subsec:results}

\begin{theorem}\label{thm:restrict to bayesian}
Suppose there exists some reference symbol $\hat{x}_0\in\hat{\mathX}$ such that $ \E_{P_{\theta}}[\int_0^T l(X_t,\hat{x}_0) dt] < \infty$ for all $\theta\in\Lambda$. Let $\mathQ$ denote the convex hull of the closure of the uncertainty set $\mathP$, i.e., $\mathQ = \conv(\cl(\{P_{\theta} ; \theta\in\Lambda\}))$. Let $l(\cdot,\cdot)$ be a loss function with properties (P1) and (P2). Then, the minimax estimator is a Bayesian estimator, i.e., 
\begin{align}
&\minimax(\Lambda)\nonumber\\
 =&\min_{Q\in\mathQ}\sup_{\theta\in\Lambda}   \regretbayesQ \nonumber\\
=&\min_{Q\in\mathQ}\sup_{\theta\in\Lambda}   \{\cmlbayesQ -\cmlbayesPtheta \}\label{eq:minimax is bayesian} .
\end{align}
\end{theorem}

Consider the following two canonical continuous-time channel models which define the conditional law $P_{Y^T|X^T}$.
\subsubsection{Gaussian Channel}\label{subsubsec:gaussian channel}
Suppose that under all $P_\theta\in\mathP$, the output process $Y^T$ is the AWGN corrupted version of $X^T$, i.e.,
\beas
dY_t = X_t dt + dW_t
\eeas
where $W^T$ is a standard Brownian motion independent of $X^T$. We consider half the squared loss function which is $l(x,\hat{x}) = \frac{1}{2}(x-\hat{x})^2$, where we introduce the factor $1/2$ to streamline the exposition that follows. 

\subsubsection{Poisson Channel}\label{subsubsec:poisson channel}
Suppose that under all $P_\theta \in \mathP$, the output $Y^T$ is a non-homogeneous Poisson process with intensity $X^T$, where $X^T$ is a nonnegative stochastic process. As in \cite{atar12}, we employ the \emph{natural loss function} $l(x,\hat{x})=x\log(x/\hat{x})-x+\hat{x}$. This loss function is a natural choice for the Poisson channel, cf. \cite[Lemma 2.1]{atar12}.

 Let define a virtual channel which takes $\theta\in\Lambda$ as an input. The corresponding output of the virtual channel is $Y^T$ which is a realization of the output process when the input has law $P_{\theta}$. Then the capacity of the virtual channel is $\sup_{w\in\mu(\Lambda)} I(\Theta;Y^T)$ where $\Theta$ is a random variable that takes a value from $\Lambda$ and $\mu(\Lambda)$ denotes the class of all probability measures on the set $\Lambda$. We are now ready to state our main results.

\begin{theorem}[Regret-Capacity]\label{thm:minimax to mutual information}
 Let the setting be either that of the Gaussian channel or the Poisson channel. Then $\minimax(\Lambda)$ is equal to the capacity of the virtual channel, i.e., 
\begin{align}
\minimax(\Lambda) =\sup_{w\in\mu(\Lambda)} I(\Theta;Y^T). \label{eq:minimax to mutual information}
\end{align}
\end{theorem}

\begin{theorem}[Minimax Filter]\label{thm:minimax filter}
Suppose the supremum in \eqref{eq:minimax to mutual information} is achieved by $w^*\in\mu(\Lambda)$. Then the minimum in \eqref{eq:minimax is bayesian} is achieved by the Bayesian optimum filter with respect to $Q^*$ where $Q^*$ is the mixture of $P_{\theta}$'s with respect to $w^*$, i.e.,
\beas
Q^* = \int_{\theta\in\Lambda} P_{\theta} w^*(d\theta).
\eeas
Moreover, the minimax filter is $\bayesQstar$.
\end{theorem}

\begin{theorem}[Strong Regret-Capacity]\label{thm:strong regret-capacity}
Suppose the supremum in \eqref{eq:minimax to mutual information} is achieved by $w^*\in\mu(\Lambda)$. For any filter $\hat{X}$ and every $\epsilon>0$,
\beas
\regretxhat  > (1-\epsilon)\cdot \minimax(\Lambda)
\eeas
for all $\theta\in\Lambda$ with the possible exception of points in a subset $B\subset \Lambda$, where
\beas
w^*(B) \leq e \cdot 2^{-\epsilon\cdot\minimax(\Lambda)}.
\eeas
\end{theorem}

Consider the case of the presence of feedback where $X_t$ is also affected by previous output $\{Y_s:0\leq s<t\}$. More precisely, $X_t$ can be viewed as a function of $Y^{t-\delta}$ and $U$ for some $\delta>0$ where $U$ is an additional randomness independent of all other processes. Let $\mathP$ be a class of joint laws of $(X^T,Y^T)$ and $\Lambda$ be a set of indices of laws. Let the definition of $\minimax(\Lambda)$ and $\regretbayesQ$ remain the same. Then, the following theorem tells us that all the above results hold essentially verbatim, i.e.,
\begin{theorem}[Presence of Feedback]\label{thm:presence of feedback} 
\begin{align*}
\minimax(\Lambda) =&\min_{Q\in\mathQ}\sup_{\theta\in\Lambda}   \regretbayesQ.
\end{align*}
Moreover, if the setting is either Gaussian or Poisson, then
\begin{align*}
\minimax(\Lambda) &= \min_{Q\in\mathcal{Q}} \sup_{\theta\in\Lambda} \regretbayesQ\\
&=\sup_{w\in\mu(\Lambda)} I(\Theta;Y^T)\\
&=\sup_{w\in\mu(\Lambda)} I(X^T\rightarrow Y^T) - I(X^T\rightarrow Y^T|\Theta)
\end{align*}
\end{theorem}
where $I(X^T\rightarrow Y^T)$ is the directed information from $X^T$ to $Y^T$, as introduced in \cite{weissman12directedinformation}. Directed information in continuous-time is also precisely defined in Section \ref{subsubsec:directed information}.

\subsection{Discussion}\label{subsec:discussion}
Theorem \ref{thm:restrict to bayesian} implies that the minimax filter is a Bayesian filter under some law $Q$. Furthermore, this minimax optimal $Q$ is a mixture of $P_{\theta}$'s. Therefore, in order to find the minimax filter, it is enough to restrict the search space to that of Bayesian filters. This is equivalent to finding an optimum prior $Q^*$, or optimum weights $w^*$ over laws $\{P_{\theta}\}$. Note that we have not assumed anything on the statistics of the input and output processes but only the aforementioned properties of the loss function $l(\cdot,\cdot)$. 

If we assume that the noise corruption mechanism is either Gaussian or Poisson, Theorem \ref{thm:minimax to mutual information} implies that the minimax regret coincides with the capacity of the virtual channel. We present the parallel results from universal coding in Section \ref{subsubsec:redundancy capacity theory}. Furthermore, Theorem \ref{thm:minimax filter} provides a prescription for such a filter in cases. Note that the mutual information $I(\Theta;Y^T)$ is equal to $I(X^T;Y^T)-I(X^T;Y^T|\Theta)$ (since $\Theta-X^T-Y^T$ forms a Markov chain) where the first term is the mutual information between input and output when the input distribution is $Q=\int_{\theta} P_{\theta} w(d\theta)$. If the uncertainty set is a class of deterministic laws (e.g., each $\theta$ corresponds to a Dirac measure concentrated at some signal $x^T$ that satisfies the input constraints of the channel) then the right hand side of \eqref{eq:minimax to mutual information} boils down to a supremum over all distributions on the set of allowable channel inputs, i.e.,
\begin{align}
\minimax(\Lambda)&= \sup_{w\in\mu(\Lambda)} I(\Theta;Y^T)\nonumber\\
&= \sup_{w\in\mu(\Lambda)} I(X^T;Y^T)-I(X^T;Y^T|\Theta)\nonumber\\
&=\sup_{P_{X^T}\in \mathQ} I(X^T;Y^T)\label{eq:thm4 discussion eq2},
\end{align}
where $\mathQ=\conv(\cl(\mathP))$. \eqref{eq:thm4 discussion eq2} follows because $X^T$ is deterministic given $\Theta$, and therefore $I(X^T;Y^T|\Theta)=0$. Note that the right hand side of \eqref{eq:thm4 discussion eq2} is the capacity of the channel whose input is constrained to lie in the uncertainty set of signals. Moreover, letting $Q^*$ denote the capacity achieving distribution, the minimax estimator is the Bayesian estimator with respect to the law $Q^*$.
More interestingly, $Q^*$ turns out to coincide with the classical notion of the least favorable prior from estimation theory. We establish this connection in Section \ref{sec:least favorable input}. These results show the strong relation between minimax estimation and channel coding problems.

In Theorem \ref{thm:strong regret-capacity}, we can see that the minimax estimator minimizes not only the worst case regret, but also regret for most $\theta\in\Lambda$, under distribution $w^*$. Cf. [4] for a discussion of the significance and implications of this result. For example, it implies that when $\Lambda$ is a compact subset of $\R^k$ and the parametrization of the input distributions  $P_\theta$ is sufficiently smooth, the minimax filter is essentially optimum not only in the worst case sense for which it was optimized, but in fact on ``most'' of the sources over all possible filters. Note that we are not restricting filters to be Bayesian. ``Most'' here means that the Lebesgue measure of the set of parameters indexing sources is vanishing, as the value of $\minimax(\Lambda)$ grows without bound. It is often the case that $\minimax(\Lambda)$ is growing without bound as $T$ increases. For example, if the uncertainty set consists of a set that constrains the possible underlying signals rather than their laws, we have seen that $\minimax(\Lambda)$ is equal to $T$ times the channel capacity, which is growing linearly with $T$.

Theorem \ref{thm:presence of feedback} implies that the above result can be extended to the case where feedback exists. Similar to \eqref{eq:thm4 discussion eq2}, if $\mathP$ is a class of deterministic laws, i.e, $X_t$ is a function of $\theta$ and previous outputs, then,
\begin{align*}
\minimax(\Lambda) &= \sup_w I(X^T\rightarrow Y^T).
\end{align*}
Recall, this is $T$ times the channel capacity in the presence of feedback. Again, if we can find the capacity achieving scheme, it will give us the minimax filter.

\section{Least Favorable Input}\label{sec:least favorable input}
In Section \ref{sec:main result}, we saw a relation between the capacity achieving prior for a virtual channel, and the minimax estimator. More precisely, the minimax estimator is the Bayesian estimator with respect to law $Q^*$, where $Q^*$ is the capacity achieving prior. In this section, we will show that $Q^*$ coincides with the ``least favorable prior" from estimation theory. This is another interesting relation between information and estimation theory.

\subsection{Notation and Definitions}\label{subsec:notations and definitions}
Suppose $\mathS$ is a class of possible input signals with corresponding index class $\Lambda$, i.e., $\mathcal{S} = \{ f_\theta \}_{\theta \in \Lambda}$. The input process $X_t$ is equal to $f_{\theta}(t)$ for some $\theta\in\Lambda$ which is unknown to the filter. Instead of the minimax criterion that we discussed thus far, we can consider the same problem in a Bayesian setting, namely where the input signal $\{X_t,0\leq t\leq T\}$ is governed by a probability law defined on $\mathS$ where estimator knows the true distribution of the source. We also assume that the channel is either Gaussian or Poisson. Define average loss, where the input prior is $Q$ and the estimator employs the optimum Bayesian filter $\E_Q[X_t|Y^t]$ as,
\begin{align*}
r_Q \stackrel{\Delta}{=}\E_Q\left[\int_0^T l(X_t,\E_Q[X_t|Y^t])dt\right].
\end{align*}
The goal is to find the least favorable input distribution $Q\in\mu(\mathS)$ which causes the greatest average loss (rather than regret). We refer to \cite[Chapter 5]{lehmann98} for a similar concept in point estimation theory. More formally, we define the least favorable prior as follows. 
\begin{definition}
A prior distribution $Q^*$ is {\it least favorable} if $r_{Q^*}\geq r_{Q}$ for all prior distributions $Q$.
\end{definition}

We define $P_{\theta}$ to be a deterministic measure such that $P_{\theta}(X_t = f_{\theta}(t) \mbox{ for all $0\leq t\leq T$}) = 1$ and consider the corresponding minimax estimation problem. Note that $\cmlbayesPtheta =0$, since the input process is deterministic under $P_\theta$, and therefore
\begin{align*}
\regretxhat  = \cmle(\theta,\hat{X}) -\cmlbayesPtheta  = \cmle(\theta,\hat{X}). 
\end{align*}
In this setting, the minimax estimator can be viewed as an achiever of $\min_{\hat{X}}\sup_{\theta\in\Lambda} \cmle(\theta,\hat{X})$. We already showed in \eqref{eq:thm4 discussion eq2} that the minimax estimator is the Bayesian estimator with respect to $Q^*$ where $Q^*$ is a capacity achieving prior.

\subsection{Relation to the Least Favorable Input}\label{subsec:in relation to least favorable input}
The relation between the minimax estimator and the least favorable input is characterized in the following theorem.

\begin{theorem}\label{thm:least favorable input}
Suppose that $Q^*$ is a distribution on $\mathS$ such that
\begin{align*}
r_{Q^*} = \sup_{\theta\in\Lambda} \cmlbayesQstar
\end{align*}
Then:
\begin{enumerate}
\item $\bayesQstar$ is a minimax estimator.
\item If $\bayesQstar$ is the unique minimizer of  $\E_{Q^*} \left[\int_0^T l(X_t,\hat{X}_t(Y^t))dt\right]$, then it is the unique minimax estimator.
\item $Q^*$ is least favorable.
\end{enumerate}
\end{theorem}

\begin{proof}
\begin{enumerate}
\item For any filter $\hat{X}$,
\begin{align}
&\sup_{\theta\in\Lambda} \cmle(\theta,\hat{X}) \nonumber\\
&\geq \int \cmle(\theta,\hat{X}) dQ^*(\theta)\nonumber\\
&= \E_{Q^*} \left[\int_0^T l(X_t,\hat{X}_t(Y^t))dt\right]\nonumber\\
&\geq \E_{Q^*} \left[\int_0^T l(X_t,\E_{Q^*}[X_t|Y^t])dt\right]\label{eq:uniquenessminimax}\\
&=r_{Q^*}\nonumber\\
&=\sup_{\theta\in\Lambda} \cmlbayesQstar.\nonumber
\end{align}
This implies 
\begin{align*}
\inf_{\hat{X}}\sup_{\theta\in\Lambda} \cmle(\theta,\hat{X}) = \sup_{\theta\in\Lambda} \cmlbayesQstar.
\end{align*}
Therefore, $\bayesQstar$ is a minimax estimator. 

\item By assumption, \eqref{eq:uniquenessminimax} holds with equality only if $\hat{X}_t(Y^t) = \E_{Q^*}[X_t|Y^t]$. This implies the uniqueness of the minimax estimator.

\item For any prior $Q$ ,
\begin{align*}
r_{Q} &= \E_{Q} \left[\int_0^T l(X_t,\E_{Q}[X_t|Y^t])dt\right]\\
&\leq \E_{Q} \left[\int_0^T l(X_t,\E_{Q^*}[X_t|Y^t])dt\right]\\
&\leq \int \cmlbayesQstar dQ(\theta)\\
&\leq \sup_{\theta\in\Lambda} \cmlbayesQstar\\
&=r_{Q^*}.
\end{align*}
This implies $Q^*$ is least favorable.
\end{enumerate}
\end{proof}
When $l(\cdot,\cdot)$ is a Bregman divergence, the minimizer of $\min_{\hat{x}}\E[l(X,\hat{x})]$ is unique, and therefore $\bayesQstar$ is the unique minimizer of  $\E_{Q^*} \left[\int_0^T l(X_t,\hat{X}_t(Y^t))dt\right]$. Furthermore, if $r_{Q^*}=\sup_{\theta\in\Lambda} \cmlbayesQstar$, then $\bayesQstar$ is the unique minimax filter. 

Theorem \ref{thm:least favorable input} provides a sufficient condition for $Q^*$ to be least favorable. Using this theorem, we can show that the least favorable input is equal to the capacity achieving prior. 
\begin{theorem}\label{thm:capacity achieving = least favorable}
If $Q^*$ is a capacity achieving prior of the channel when the input is restricted to the set $\mathS$, then $Q^*$ is a least favorable input.
\end{theorem}
\begin{proof}
Since our uncertainty set is a collection of deterministic measures, we can apply \eqref{eq:thm4 discussion eq2};
\begin{align*}
\min_{Q\in\mu(\mathS)} \sup_{\theta\in\Lambda} \cmlbayesQ  = \sup_{Q\in\mu(\mathS)} I(X^T;Y^T).
\end{align*}
Since $Q^*$ achieves both the minimum and supremum of $\min_{Q\in\mu(\mathS)} \sup_{\theta\in\Lambda} \cmlbayesQ $ and $\sup_{Q\in\mu(\mathS)} I(X^T;Y^T)$, respectively, we can write
\begin{align}
\sup_{\theta\in\Lambda} \cmlbayesQstar &= I(X^T;Y^T)\label{eq:specify law of XT}\\
&=\E_{Q^*}\left[\int_0^T l(X_t,\E_{Q^*}[X_t|Y^t] dt\right]\label{eq:using Duncan}\\
&=r_{Q^*},\nonumber
\end{align}
where the probability law of $X^T$ in \eqref{eq:specify law of XT} is $Q^*$. Line \eqref{eq:using Duncan} is due to the relation between mutual information and the causal estimation loss. C.f. \cite{duncan70} and \cite{atar12} for Gaussian and Poisson cases respectively. This result tells us that $Q^*$ satisfies the condition of Theorem \ref{thm:least favorable input}, and therefore the capacity achieving prior $Q^*$ is least favorable.
\end{proof}

\subsection{Examples}\label{subsec:least favorable input example}
We have shown that the least-favorable prior and the capacity-achieving prior always coincide in continuous-time \emph{causal} estimation. However, this may not be true in general estimation problem. Consider the problem of minimax estimation of a bounded normal mean. We have a noisy observation 
\beas
Y=x+Z
\eeas
where $x\in[-a,a]$ is a bounded scalar parameter and $Z$ is a standard normal random variable. We can consider the least favorable input in this setting. The least favorable input is simply defined by $\arg\max_Q \E_Q[(X-\E_Q[X|Y])^2]$ where the maximum is over all probability laws of $X$ on $[-a,a]$. For $a>1.05$, the unique least favorable prior is supported on at least 3 discrete points in $[-a,a]$ \cite{casella81}. 

On the other hand, consider the corresponding peak power constrained Gaussian channel capacity problem:
\begin{align*}
\sup_{P_X\in\mu([-a,a])} I(X;X+Z).
\end{align*}
Sharma and Shamai showed that $P_X^*=\frac{1}{2}\delta_{-a}+\frac{1}{2}\delta_a$ achieves capacity for all $a\leq 1.671$ (\cite{sharma08}, \cite{sharma10}). Therefore the least favorable prior and the capacity achieving distributions do not coincide when $1.05<a<1.671$. This example shows that the least favorable prior and the capacity achieving distribution do not coincide in general.

Now let us examine an analogous but contrasting continuous-time causal estimation problem. Consider the input process $X_t \equiv x$ for all $0\leq t\leq T=1$, where $x\in[-a,a]$ is a bounded scalar parameter and $a>0$. We observe $Y^T$, the output of AWGN channel $dY_t = X_t dt+dW_t$. In this setting, the least favorable input can be defined by $\arg\max_Q\E_Q\left[ \int_0^T (X-\E_Q[X_t|Y^t])^2dt\right]$ where the maximum is over all probability laws of $X$ on $[-a,a]$. 

On the other hand, the corresponding channel capacity problem remains the same, i.e.,
\begin{align*}
 \sup_{Q\in\mu([-a,a])} I(X^T;Y^T) = \sup_{Q\in\mu([-a,a])} I(X;Y_T).
\end{align*}
Theorem \ref{thm:capacity achieving = least favorable} tells us that the least favorable prior coincides with the capacity achieving prior. Therefore, both the capacity achieving prior and the least favorable prior are $Q^* = \frac{1}{2}\delta_{-a}+\frac{1}{2} \delta_a$ if $a\leq 1.671$.

\section{Proof}\label{sec:proof}

\subsection{Preliminaries}\label{subsec:preliminaries}
\subsubsection{Redundancy Capacity Theory}\label{subsubsec:redundancy capacity theory}
It is worth reviewing some results from universal source coding theory, since the techniques will be useful in proving some of our results. In the context of universal source coding, let $x^n = (x_1,\cdots,x_n)$ be a sequence of symbols. Let $\{P_{\theta}:\theta\in\Lambda\}$ be a set of probability laws of sequences. Define redundancy by
\beas
R_n(L,\theta) = \E_{P_{\theta}}[L(X^n)]-H_{\theta}(X^n)
\eeas
where $L(X^n)$ is length of codewords for given uniquely decodable (UD) code and $H_{\theta}(X^n)$ is an entropy of sequence with respect to $P_{\theta}$. Then, define minimax redundancy as
\beas
R_n = \min_{L}\sup_{\theta\in\Lambda} R_n(L,\theta).
\eeas

In \cite{gallager76}, Gallager showed that minimax redundancy is equal to the capacity of the virtual channel, where its input is $\theta\in\Lambda$ and output is drawn by probability measure $P_{\theta}(x^n)$, i.e.,
\beas
R_n = C_n
\eeas
where $C_n = \sup_w I (\Theta;X^n)$ and the supremum is over all priors of random variable $\Theta$ on $\Lambda$.

Furthermore, the minimum achieving length function $L^*$ is related to the supremum achieving weights $w^*$ in the following manner: 
\beas
L^*(x^n) = -\log Q^*(x^n)
\eeas
where $Q^* = \int_{\theta\in\Lambda} P_{\theta} w^*(d\theta)$.

Merhav and Feder \cite{merhav95} proved the strong version of redundancy capacity theorem which is for any length function $L$ of a UD code and every $\epsilon>0$,
\beas
R_n(L,\theta) > (1-\epsilon)C_n,
\eeas
for all $\theta\in\Lambda$ except for points in a subset $B\subset\Lambda$ where
\bea
w^*(B)\leq e\cdot 2^{-\epsilon C_n}.\label{eq:strong redundancy capacity}
\eea
In \eqref{eq:strong redundancy capacity}, the choice of probability measure $w^*$ is reasonable because it captures variety in sets (cf. Merhav and Feder \cite{merhav95}). This theorem implies that $L^*$ is not only the minimum of worst case redundancy, but also close to minimum redundancy for most of other sources.

\subsubsection{Directed Information}\label{subsubsec:directed information}
Given two random vectors $X^n$ and $Y^n$, we can define directed information.
\begin{definition}[Discrete-time Directed Information]
\beas
I(X^n\rightarrow Y^n) \triangleq \sum_{i=1}^n I(X^i;Y_i|Y^{i-1}).
\eeas
\end{definition}

In \cite{weissman12directedinformation}, Weissman et al. extended this definition to the continuous-time setting, i.e., directed information between two random processes $X^T$ and $Y^T$. For given vector ${\bf t} = (t_0,\cdots,t_n)$ where $0=t_0<t_1<\cdots<t_n=T$, define $ X_0^{T,{\bf t}} \triangleq (X_0^{t_1},X_{t_1}^{t_2},\cdots,X_{t_{n-1}}^T)$ and treat $X_0^{T,{\bf t}}$ as a $n$ dimensional vector. Using this notation, we can define the directed information between two random processes.

\begin{definition}[Continuous-time Directed Information]
\beas 
I(X^T\rightarrow Y^T) \triangleq \inf_{\bf t} I(X_0^{T,{\bf t}} \rightarrow Y_0^{T,{\bf t}}) 
\eeas
where the infimum is over all finite dimensional vectors ${\bf t}$.
\end{definition}

We refer to \cite{weissman12directedinformation} for more on the properties of directed information and its significance in communication and estimation. 

\subsection{Proof of Theorem \ref{thm:restrict to bayesian}}\label{subsec:proof of theorem 1}
\begin{proof}
We denote the class of measures on $\Lambda$ by $\mu(\Lambda)$, i.e., $w\in\mu(\Lambda)$ can be viewed as a weight function of each probability distribution in $P_{\theta}$ where $\theta\in\Lambda$. Then we have
\begin{align*}
&\minimax(\Lambda)\\
 &=\inf_{\hat{X}} \sup_{\theta\in\Lambda} \regretxhat \\
&=\inf_{\hat{X}} \sup_{\theta\in\Lambda}\left\{  \cmle(\theta,\hat{X}) - \cmlbayesPtheta \right\}\\
&=\inf_{\hat{X}} \sup_{w\in\mu(\Lambda)}\left\{\int_{\theta\in\Lambda} \left( \cmle(\theta,\hat{X}) - \cmlbayesPtheta \right)w(d\theta)\right\}.
\end{align*}
let $P_{av} = \int P_{\theta} w(d\theta)$. Use Fubini's theorem; since there exists some reference symbol $\hat{x}_0\in\hat{\mathX}$ such that $ \E_{P_{\theta}}[\int_0^T l(X_t,\hat{x}_0) dt] < \infty$ for all $\theta\in\Lambda$, there exists a filter $\hat{X}$ such that $\int_0^T l(X_t,\hat{X}_t(Y^t))dt$ is $L^1$ with respect to all $P_{\theta}$. Therefore,
\begin{align*}
\int_{\theta\in\Lambda}  \cmle(\theta,\hat{X}) w(d\theta) = \E_{P_{av}}\left[\int_0^T l(X_t,\hat{X}_t(Y^t))dt \right].
\end{align*}

\begin{figure*}[!t]
\normalsize
\setcounter{MYtempeqncnt}{\value{equation}}
\begin{align}
\minimax(\Lambda)&=\inf_{\hat{X}}\sup_{w\in\mu(\Lambda)}\left\{\E_{P_{av}}\left[\int_0^T l(X_t,\hat{X}_t(Y^t))dt \right] - \int_{\theta\in\Lambda} \cmlbayesPtheta w(d\theta)\right\}\label{eq:theoremproof1}\\
&\geq \sup_{w\in\mu(\Lambda)}\inf_{\hat{X}} \left\{\E_{P_{av}}\left[\int_0^T l(X_t,\hat{X}_t(Y^t))dt \right] - \int_{\theta\in\Lambda} \cmlbayesPtheta w(d\theta)\right\}\label{eq:theoremproof2}\\
&= \sup_{w\in\mu(\Lambda)} \left\{\E_{P_{av}}\left[\int_0^T l(X_t,\E_{P_{av}}[X_t|Y^t])dt \right] - \int_{\theta\in\Lambda} \cmlbayesPtheta w(d\theta)\right\}\label{eq:theoremproof3}\\
&= \sup_{w\in\mu(\Lambda)}\min_{Q\in\mathQ} \left\{\E_{P_{av}}\left[\int_0^T l(X_t,\E_{Q}[X_t|Y^t])dt \right] - \int_{\theta\in\Lambda} \cmlbayesPtheta w(d\theta)\right\}\nonumber\\
&= \min_{Q\in\mathQ}\sup_{w\in\mu(\Lambda)} \left\{\E_{P_{av}}\left[\int_0^T l(X_t,\E_{Q}[X_t|Y^t])dt \right] - \int_{\theta\in\Lambda}\cmlbayesPtheta w(d\theta)\right\}\label{eq:theoremproof4}\\
&= \min_{Q\in\mathQ}\sup_{w\in\mu(\Lambda)} \left\{\int_{\theta\in\Lambda} \left(\E_{P_{\theta}}\left[\int_0^T l(X_t,\E_{Q}[X_t|Y^t])dt \right] -  \cmlbayesPtheta \right)w(d\theta)\right\}\label{eq:theoremFubini2}\\
&= \min_{Q\in\mathQ}\sup_{\theta\in\Lambda} \left\{\E_{P_{\theta}}\left[\int_0^T l(X_t,\E_{Q}[X_t|Y^t])dt \right] -  \cmlbayesPtheta \right\}\nonumber\\
&= \min_{Q\in\mathQ}\sup_{\theta\in\Lambda} \left\{\cmlbayesQ  -  \cmlbayesPtheta \right\}\nonumber\\
&=\min_{Q\in\mathQ}\sup_{\theta\in\Lambda}   \regretbayesQ.\nonumber
\end{align}
\hrulefill
\vspace*{4pt}
\end{figure*}
The remaining proof of 
\begin{align*}
\minimax(\Lambda)\geq \min_{Q\in\mathQ}\sup_{\theta\in\Lambda}   \regretbayesQ
\end{align*} appears at the top of the next page where:
\begin{itemize}
\item \eqref{eq:theoremproof2} is because for any real-valued function $f(x,y)$ on $\mathX\times\mathY$, we have 
\begin{align*}
\inf_{x\in\mathX}\sup_{y\in\mathY} f(x,y)\geq\sup_{y\in\mathY}\inf_{x\in\mathX} f(x,y).
\end{align*}
\item \eqref{eq:theoremproof3} is because the loss function $l$ satisfies property (P2) (expectation minimizes the loss function). 
\item \eqref{eq:theoremproof4} is because of Sion's minimax theorem. In order to apply Sion's minimax theorem, we have to show the following four conditions;
\begin{itemize}
\item $\mathQ$ has to be a compact convex subset of a linear topological space 
\item $\mu(\Lambda)$ has to be a convex subset of a linear topological space
\item We have to show that
\begin{align*}
\E_{P_{av}}&\left[\int_0^T l(X_t,\E_{Q}[X_t|Y^t])dt \right]\\ &- \int_{\theta\in\Lambda}\cmlbayesPtheta w(d\theta)
\end{align*} 
is upper semi-continuous and quasiconcave on $\mu(\Lambda)$ for all $Q\in \mathQ$.
\item We also have to show that
\begin{align*}
\E_{P_{av}}&\left[\int_0^T l(X_t,\E_{Q}[X_t|Y^t])dt \right]\\& - \int_{\theta\in\Lambda}\cmlbayesPtheta w(d\theta)
\end{align*} 
is lower semi-continuous and quasi-convex on $\mathQ$ for all $w\in\mu(\Lambda)$.
\end{itemize}
Consider the topology of weak convergence of probability laws. Since $\mathP=\{P_{\theta}:\theta\in\Lambda\}$ is tight and $\mathX$ is a Polish space, we can apply Prohorov's theorem which implies that the closure of $\mathP$ is compact. Since convex hull of compact set is always compact, and therefore $\mathQ$ is compact. Convexity of $\mu(\Lambda)$ and upper semi-continuity are clear. Lower semi-continuity is clear since we assumed that $l(\cdot,\cdot)$ is a lower semi-continuous in the second argument. This guarantees that 
\begin{align*}
\E_{P_{av}}&\left[\int_0^T l(X_t,\E_{Q}[X_t|Y^t])dt \right]\\& - \int_{\theta\in\Lambda}\cmlbayesPtheta w(d\theta)
\end{align*}
 is lower semi-continuous in $Q\in\mathQ$. 
\item Note that \eqref{eq:theoremFubini2} also holds due to a similar argument with \eqref{eq:theoremproof1}.
\end{itemize}

The opposite direction is trivial, that is
\begin{align*}
&\inf_{\hat{X}} \sup_{\theta\in\Lambda}\left\{ \cmle(\theta,\hat{X}) - \cmlbayesPtheta \right\}\\
&\leq \min_{Q\in\mathQ}\sup_{\theta\in\Lambda} \left\{ \cmlbayesQ  - \cmlbayesPtheta \right\}.
\end{align*}
Therefore,
\begin{align*}
\minimax(\Lambda) =\inf_{\hat{X}} \sup_{\theta\in\Lambda} \regretxhat =\min_{Q\in\mathQ}\sup_{\theta\in\Lambda}   \regretbayesQ.
\end{align*}
\end{proof}

\subsection{Proof of Theorems \ref{thm:minimax to mutual information} and \ref{thm:minimax filter}}\label{proof of theorem 2,3}
\begin{proof}
For both Gaussian and Poisson setting, the cost of mismatch is related to relative entropy between outputs corresponding to input laws $P_{\theta}$ and $Q$, respectively \cite{weissman10}, \cite{atar12}. In other words, if $(P_{\theta})_{Y^T}$ is the distribution of $Y^T$ where the law of the input process is $P_{\theta}$, and if $Q_{Y^T}$ is defined similarly, we have 
\bea
\cmlbayesQ  - \cmlbayesPtheta = D((P_{\theta})_{Y^T}||Q_{Y^T}).
\eea
Using similar argument from classical minimax redundancy theory, we can get
\begin{align}
&\minimax(\Lambda)\nonumber\\
 =&\min_{Q\in\mathQ}\sup_{\theta\in\Lambda}  \{\cmlbayesQ  - \cmlbayesPtheta \} \nonumber\\
=&\min_{Q\in\mathQ}\sup_{\theta\in\Lambda}  D((P_{\theta})_{Y^T}||Q_{Y^T})\nonumber\\
=&\min_{Q\in\mathQ}\sup_{\theta\in\Lambda}  \int d(P_{\theta})_{Y^T}\log\left(\frac{d(P_{\theta})_{Y^T}}{dQ_{Y^T}}\right)\nonumber\\
=&\min_{Q\in\mathQ}\sup_{w\in\mu(\Lambda)}  \int\int d(P_{\theta})_{Y^T}\log\left(\frac{d(P_{\theta})_{Y^T}}{dQ_{Y^T}}\right)w(d\theta)\nonumber\\
=&\sup_{w\in\mu(\Lambda)}\min_{Q\in\mathQ}  \int\int d(P_{\theta})_{Y^T}\log\left(\frac{d(P_{\theta})_{Y^T}}{dQ_{Y^T}}\right)w(d\theta)\label{eq:thm2:minimax}\\
=&\sup_{w\in\mu(\Lambda)}\min_{Q\in\mathQ}  \int\int d(P_{\theta})_{Y^T}\log\left(\frac{d(P_{\theta})_{Y^T}}{d(P_{av})_{Y^T}}\right)w(d\theta)\nonumber\\
&+\int\int d(P_{\theta})_{Y^T}\log\left(\frac{d(P_{av})_{Y^T}}{dQ_{Y^T}}\right)w(d\theta)\nonumber\\
=&\sup_{w\in\mu(\Lambda)}\min_{Q\in\mathQ}  \int D((P_{\theta})_{Y^T}||(P_{av})_{Y^T})w(d\theta)\nonumber\\
&+D((P_{av})_{Y^T}||Q_{Y^T})\nonumber\\
=&\sup_{w\in\mu(\Lambda)} \int D((P_{\theta})_{Y^T}||(P_{av})_{Y^T})w(d\theta)\label{eq:proof of theorem2 eq2}\\
=&\sup_{w\in\mu(\Lambda)} I(\Theta;Y^T).\nonumber
\end{align}
In \eqref{eq:thm2:minimax}, we applied the minimax theorem again where weak lower semi-continuity in $Q$ follows from the property of the relative entropy. All other conditions for minimax theorem are the same as the proof in the previous section. This completes the proof of Theorem \ref{thm:minimax to mutual information}.

In \eqref{eq:proof of theorem2 eq2}, if a supremum achieving $w^*$ exists, the minimum achieving $Q^*$ is $P_{av}$, i.e.,
\beas
Q^* = \int_{\theta\in\Lambda} P_{\theta} w^*(d\theta).
\eeas
Therefore, 
\beas
\minimax(\Lambda) = \sup_{\theta\in\Lambda} \{\cmlbayesQstar - \cmlbayesPtheta \},
\eeas
which implies the minimax estimator is a Bayesian estimator based on law $Q^*$, i.e.,
\beas
\hat{X}(Y^t) = \E_{Q^*}[X_t|Y^t].
\eeas
\end{proof}

\subsection{Proof of Theorem \ref{thm:strong regret-capacity}}\label{subsec:proof of theorem 4}
\begin{proof}
The idea of proof is similar to that in \cite{merhav95}. For given estimator $\hat{X}^*$ and $\epsilon>0$, define the set $B = \{\theta : R(\theta,\hat{X}^*)\leq (1-\epsilon) \cdot\minimax(\Lambda)\}$. Then, by definition of $B$, we have
\begin{align*}
\minimax(B) &= \inf_{\hat{X}} \sup_{\theta\in B} \regretxhat \\
&\leq \sup_{\theta\in B} R(\theta,\hat{X}^*)\\
&\leq (1-\epsilon) \cdot \minimax(\Lambda).
\end{align*}

Consider $\Theta$ as a random variable with measure $w^*$ where $w^*$ achieves the supremum of \eqref{eq:minimax to mutual information}. Let $Z=\mathbf{1}_{\{\Theta\in B\}}$ be a binary random variable. Clearly we have $P(Z=1) = w^*(B)$. Since $Z - \Theta - Y^T$ is a Markov chain, we have
\begin{align}
\minimax(\Lambda) =& I(\Theta;Y^T)\nonumber\\
=&I(Z;Y^T)+I(\Theta;Y^T|Z)\nonumber\\
=&I(Z;Y^T)+P(Z=1)I(\Theta;Y^T|Z=1)\nonumber\\
&+P(Z=0)I(\Theta;Y^T|Z=0)\nonumber\\
\leq& I(Z;Y^T)+w^*(B) \cdot\minimax(B)\nonumber\\
& + (1-w^*(B))\cdot\minimax(\Lambda)\label{eq:proof thm4 eq1}\\
\leq& H(Z) + (1-\epsilon\cdot w^*(B))\cdot\minimax(\Lambda)\nonumber.
\end{align}
\eqref{eq:proof thm4 eq1} is because $I(\Theta;Y^T|Z=1) = \minimax(B)$ and $I(\Theta;Y^T|Z=0)\leq \minimax(\Lambda)$.
Finally, we get 
\beas
&-\log w^*(B) - \frac{1-w^*(B)}{w^*(B)}\log(1-w^*(B)) \nonumber\\
&\geq \epsilon\cdot\minimax(\Lambda),
\eeas
which implies
\beas
w^*(B) \leq e\cdot 2^{-\epsilon\cdot\minimax(\Lambda)}.
\eeas
\end{proof}

\subsection{Proof of Theorem \ref{thm:presence of feedback}}\label{subsec:proof of theorem 5}
\begin{proof}
Proofs of Theorem \ref{thm:restrict to bayesian} and \ref{thm:strong regret-capacity} are still valid even with a feedback. Moreover, since the result of cost of mismatch also valid with feedback \cite{atar12}, the only non-trivial part is to show $I(\Theta;Y^T) = I(X^T\rightarrow Y^T) - I(X^T\rightarrow Y^T|\Theta)$.

Recall the definition of directed information in continuous-time setting. For fixed time intervals $0=t_0<t_1<t_2<\cdots<t_n=T$,
\begin{align}
I(\Theta;Y^T) =& \sum_{i=1}^n I(\Theta; Y^{t_i}_{t_{i-1}}|Y^{t_{i-1}})\nonumber\\
=& \sum_{i=1}^n \int \log\frac{dP_{Y^{t_i}_{t_{i-1}}|Y^{t_{i-1}},\Theta}}{dP_{Y_{t_{i-1}}^{t_i}|Y^{t_{i-1}}}} dP_{Y^{t_i},\Theta}\nonumber\\
=& \sum_{i=1}^n \int \log\frac{dP_{Y^{t_i}_{t_{i-1}}|X^{t_i},Y^{t_{i-1}},\Theta}}{dP_{Y^{t_i}_{t_{i-1}}|Y^{t_{i-1}}}}\nonumber\\
&- \log\frac{dP_{Y^{t_i}_{t_{i-1}}|X^{t_i},Y^{t_{i-1}},\Theta}}{dP_{Y^{t_i}_{t_{i-1}}|Y^{t_{i-1}},\Theta}} dP_{X^i,Y^{t_i},\Theta}\nonumber\\
=& \sum_{i=1}^n \int \log\frac{dP_{Y^{t_i}_{t_{i-1}}|X^{t_i},Y^{t_{i-1}}}}{dP_{Y^{t_i}_{t_{i-1}}|Y^{t_{i-1}}}}dP_{X^i,Y^{t_i}}\nonumber\\
&- \int\log\frac{dP_{Y^{t_i}_{t_{i-1}}|X^{t_i},Y^{t_{i-1}},\Theta}}{dP_{Y^{t_i}_{t_{i-1}}|Y^{t_{i-1}},\Theta}} dP_{X^{t_i},Y^{t_i},\Theta}\label{eq:markovity in directed information}\\
=& \sum_{i=1}^n I(Y^{t_i}_{t_{i-1}};X^{t_i}|Y^{t_{i-1}}) - I(Y_i;X^{t_i}|Y^{t_{i-1}},\Theta),\nonumber
\end{align}
where \eqref{eq:markovity in directed information} is because $\Theta~-~(X^{t_i},Y^{t_{i-1}})~-~Y^{t_i}_{t_{i-1}}$ forms a Markov chain. Since the equality holds for any choice of time intervals, we take ${\bf t}$'s such that $\sup_i ||t_i-t_{i-1}|| \rightarrow 0$ and conclude
\begin{align*}
\minimax(\Lambda) &= \min_{Q\in\mathcal{Q}} \sup_{\theta\in\Lambda} \regretbayesQ\\
&=\min_{Q\in\mathcal{Q}} \sup_{\theta\in\Lambda} D((P_{\theta})_{Y^T}||Q_{Y^T})\\
&=\sup_{w} I(\Theta;Y^T)\\
&=\sup_w I(X^T\rightarrow Y^T) - I(X^T\rightarrow Y^T|\Theta).
\end{align*}
\end{proof}

\section{Examples}\label{sec:examples}
\subsection{Gaussian Channel and Sparse Signal}\label{subsec:gaussian channel and sparse signal(examples)}
We first apply our theorems to the problem of sparse signal estimation under Gaussian noise.

\subsubsection{Setting}\label{subsubsec:setting}
 We assume output process $Y^T$ is an AWGN corrupted version of $X^T$ as we discussed in Section \ref{subsubsec:gaussian channel}. The input process $X^T$ is sparse (the meaning  will be explained). Recall that we are using half of a mean squared error as a distortion measure, $l(x,\hat{x}) = \frac{1}{2}(x-\hat{x})^2$.

Let $\{\phi_i(t),0\leq t\leq T\}_{i=1}^n$ be a given orthonormal signal set which is known to the estimator. Suppose $X^T$ is a linear combination of $\phi_i(t)$'s, i.e., $X_t = \sum_{i=1}^n A_i \phi_i(t)$ where $\{A_i\}_{i=1}^n$ are random variables with unknown distribution. However, we assume that the estimator knows that the signal $X^T$ is power constrained and is sparse, by which we mean  that the fraction of nonzero elements in $\{A_i\}_{i=1}^n$ should be smaller than $q$ (i.e., at most $nq$ number of $A_i$'s can be nonzero). Let $\mathP$ be a class of all probability measures $P_{\theta}$ of vector $A = (A_1,\cdots,A_n)$ indexed by $\theta$ which satisfies these two constraints almost surely, i.e.,
\begin{align}
\mathP = \left\{P_{\theta} : \frac{1}{n}\sum_{i=1}^n A_i^2  \leq P,\frac{1}{n}\sum_{i=1}^n \mathbf{1}_{\{A_i\neq0\}}\leq q\mbox{ a.s.}\right\}. \label{eq:sparse signal set}
\end{align}
Note that $\int_0^T X_t^2 dt = \sum_{i=1}^n A_i^2$ because of orthonormality, and therefore it is equivalent to consider $\frac{1}{n}\sum_{i=1}^n A_i^2 \leq P$ as the power constraint. Define an uncertainty set $\Lambda$ be the set of such indices. It is clear that $\mathP=\{P_{\theta}:\theta\in\Lambda\}$ is a convex set.

We further define $\mathP_D$ as a class of deterministic measures $P_{\theta}\in\mathP$ (i.e., $P_{\theta}(\{a^n\})=1$ for some $a^n\in\R^n$), and the corresponding set of indices as $\Lambda_D$. Note that $\conv(\mathP_D)=\mathP$. We also define the class of sparse signals with average constraints
\beas
\mathP_{av} = \left\{P_{\theta} : \E\left[\frac{1}{n}\sum_{i=1}^n A_i^2\right]  \leq P,\E\left[\frac{1}{n}\sum_{i=1}^n \mathbf{1}_{\{A_i\neq0\}}\right]\leq q\right\}. 
\eeas
and the corresponding index set $\Lambda_{av}$.

We can understand $\mathP_D$ as a class of Dirac measures at some $a^n$, and $\mathP_{av}$ as a class of measures that satisfy average power and sparsity constraints in expectation while measures in $\mathP$ satisfies constraints with probability 1. In classical minimax statistical theory, $\mathP_D$ is often called the set of point uncertainty, and $\mathP_{av}$ is called minimax Bayes relaxation. Also, define the corresponding set of indices as $\Lambda_D$ and $\Lambda_{av}$, respectively. There are some simple relations among these sets.
\begin{itemize}
\item $\mathP_D\subset\mathP\subset\mathP_{av}$ and $\Lambda_D\subset\Lambda\subset\Lambda_{av}$
\item $\mathP$ is a convex closure of $\mathP_D$, i.e., $\mathP = \conv(\mathP_D)$.
\end{itemize}
The goal is to find $\minimax(\Lambda)$ and the minimax filter that achieves it. 

A similar {\it non-causal} minimax problem was studied by Pinsker \cite{pinsker80}. Pinsker considered the non-causal estimation problem with only the power constraint. Although Pinsker's approach does not directly apply to our setting because of the difference between non-causal and causal estimation, we will use a similar idea to argue that the approximated version of the minimax filter works well.

\subsubsection{Application of the Theorem}\label{subsubsec:apply the theorem}
It is easy to show that $\mathP$, $\mathP_D$ and $\mathP_{av}$ are tight, and therefore we can apply the theorems. Theorem \ref{thm:minimax to mutual information} implies that 
\begin{align*}
\minimax(\Lambda) = \sup_{w(\cdot)\in\mu(\Lambda)} I(X^T;Y^T)-I(X^T;Y^T|\Theta).
\end{align*}
Since our optimum causal minimax estimator is a Bayesian estimator under the distribution $Q^* = \int P_{\theta} w^*(d\theta)$ where $w^*$ achieves the supremum, we are interested in $w^*$. Rather than maximizing the difference between mutual informations, we can find an equivalent problem which is much easier to handle by exploiting the relation between $\minimax(\Lambda)$ and $\minimax(\Lambda_D)$.

\begin{lemma}\label{lem:theta_d is equal to theta}
\begin{align*}
\minimax(\Lambda_D) = \minimax(\Lambda).
\end{align*}
\end{lemma}
Appendix \ref{sec:proof of theta_d is equal to theta lemma} is dedicated to the proof of Lemma \ref{lem:theta_d is equal to theta}. Since $\mathP_D$ is a set of deterministic measures, we can get more explicit formula of $\minimax(\Lambda_D)$ as we showed in Section \ref{subsec:discussion}, 
\begin{align}
\minimax(\Lambda) &= \minimax(\Lambda_D) \nonumber\\
&= \sup_{w(\cdot)\in\mu(\Lambda_D)} I(X^T;Y^T)\label{eq:lemma lambda_D = lambda eq1}\\
&= \sup_{P_{\theta}\in\mathP} I(X^T;Y^T).\label{eq:lemma lambda_D = lambda eq2}
\end{align}
In \eqref{eq:lemma lambda_D = lambda eq1}, $X^T$ is governed by the law $\int P_{\theta} w(d\theta)$ which is an element of $\mathP$. Therefore, finding a supremum achiever $w^*$ in \eqref{eq:lemma lambda_D = lambda eq1} is equivalent to find the maximum prior $P_{\theta}^*$ in $\mathP$, thus, \eqref{eq:lemma lambda_D = lambda eq2} holds. Moreover, the minimum achiever $Q^*$ of $\minimax(\Lambda_D)$ coincides with that of $\minimax(\Lambda)$. Thus, it is enough to consider $\minimax(\Lambda_D)$ which is much simpler to solve.

Now, consider the $\minimax(\Lambda_{av})$.
\begin{align}
&\minimax(\Lambda) \nonumber\\
&= \min_{Q\in\mathP}\sup_{\theta\in\Lambda} \cmlbayesQ -\cmlbayesPtheta \nonumber\\
&=\min_{Q\in\mathP_{av}}\sup_{\theta\in\Lambda} \cmlbayesQ -\cmlbayesPtheta \label{eq:universal optimality}\\
&\leq\min_{Q\in\mathP_{av}}\sup_{\theta\in\Lambda_{av}} \cmlbayesQ -\cmlbayesPtheta \nonumber\\
&=\minimax(\Lambda_{av})\nonumber\\
&=\sup_{w(\cdot)\in\mu(\mathP_{av})} I(X^T;Y^T)-I(X^T;Y^T|\Theta)\nonumber
\end{align}
where \eqref{eq:universal optimality} is because Bayesian estimator with prior $Q^*\in\mathP$ is optimum over all possible filters and we can always extend the search space. We will use this relation between $\minimax(\Lambda)$ and $\minimax(\Lambda_{av})$ to approximate the minimax filter.

\subsubsection{Sufficient Statistics}\label{subsubsec:sufficient statistics}
Since the channel input signal is a linear combination of orthonormal signals, sufficient statistics of the channel output signal at time $t=T$ are projections on each $\phi_i$'s, i.e., $\{\int_0^T \phi_i(t)dY_t \}_{i=1}^n$. Therefore, the above mutual information $I(X^T;Y^T)$ can be further simplified as
\beas
\minimax(\Lambda) = \sup_{P_{\theta}\in\mathP} I\left(A^n;B^n\right)
\eeas
where $B_i = \int_0^T \phi_i(t)dY_t$ for $1\leq i\leq n$. Since we assumed an orthonormal basis, $B^n$ can be viewed as the output of a discrete-time additive white Gaussian channel, i.e., $B_i = A_i + W_i$ where $W_i$ is i.i.d. standard Gaussian noise and independent of $A^n$. This implies that our problem of maximizing the mutual information over the continuous-time channel is equivalent to maximizing the mutual information between $n$ channel inputs and $n$ channel outputs over the discrete AWGN channel, with the input distribution constrained as in \eqref{eq:sparse signal set}.

Recall that above result shows that sufficient statistics for estimating $X_T$ given $Y^T$ are projections, i.e., $\left\{\int_0^T \phi_i(s)dY_s\right\}_{i=1}^n$, in other words, the following Markov relation holds
\begin{align*}
X_T \quad - \quad \left\{\int_0^T \phi_i(s)dY_s\right\}_{i=1}^n \quad - \quad Y^T.
\end{align*}
Since we are looking for a causal estimator, we need a similar result for time $t<T$. The following lemma shows that $\left\{\int_0^t \phi_i(s)dY_s\right\}_{i=1}^n$ are sufficient statistics for estimating $X_t$ given $Y^t$ .
\begin{lemma}\label{lem:sufficient statistics}
The following Markov relation holds for all $t\in[0,T]$,
\begin{align*}
X_t \quad - \quad \left\{\int_0^t \phi_i(s)dY_s\right\}_{i=1}^n \quad - \quad Y^t.
\end{align*}
\end{lemma}
Proof of Lemma \ref{lem:sufficient statistics} is given in Appendix \ref{sec:proof of sufficient statistics lemma}. Using this lemma, we will show that we can compute $\E[X_t|Y^t]$ easily. 

\subsubsection{Bayesian Estimator}\label{subsubsec:bayesian estimator}
Let $Q^*$ be the minimum achieving law of $\minimax(\Lambda)$ so that the optimum causal minimax estimator is a Bayesian estimator assuming the prior $Q^*$, i.e.,
\beas
\hat{X}_t(Y^t) = \E_{Q^*} [X_t|Y^t].
\eeas
This conditional expectation is hard to compute in general. However, the sufficient statistics provide us a practical implementation of the estimator.

Let us first , define a projection vector 
\begin{align*}
\mathbf{\tilde{Y}}(t) &= [\tilde{Y}_1(t),\tilde{Y}_2(t),\cdots,\tilde{Y}_n(t)]^{\mathrm{T}}
\end{align*}
where $\tilde{Y}_i(t) = \int_0^t \phi_i(s)dY_s$. The vector $\mathbf{\tilde{Y}}(t)$ indicates a projection of $Y^t$ on the basis space. Similarly, define
\begin{align*}
\mathbf{\tilde{W}}(t) &= (\tilde{W}_1(t),\tilde{W}_2(t),\cdots,\tilde{W}_n(t))^{\mathrm{T}}\\
\mathbf{\tilde{X}}(t) &= (\tilde{X}_1(t),\tilde{X}_2(t),\cdots,\tilde{X}_n(t))^{\mathrm{T}}
\end{align*}
where $\tilde{W}_i(t) = \int_0^t \phi_i(s)dW_s$ and 
\begin{align*}
\tilde{X}_i(t) =&\int_0^t \phi_i(s)X_s ds\\
=& \sum_{j=1}^n a_j\left(\int_0^t \phi_i(s)\phi_j(s) ds\right).
\end{align*}
Let further define a $n$ by $n$ matrix $\Gamma(t)$ where $\left[\Gamma(t)\right]_{i,j} = \int_0^t \phi_i(s)\phi_j(s) ds$.

Note that $\mathbf{\tilde{W}}(t)$ is Gaussian with zero mean and covariance matrix $\Gamma(t)$ since
\begin{align*}
\E[\tilde{W}_i(t)\tilde{W}_j(t)] &= \E\left[\int_0^t\int_0^t \phi_i(s)\phi_j(u)dW_sdW_u\right]\\
&=\int_0^t \phi_i(s)\phi_j(s)ds.
\end{align*}
From Lemma \ref{lem:sufficient statistics}, for fixed $t$, the causal estimation problem is reduced to the following vector estimation problem
\begin{align*}
\mathbf{\tilde{Y}}(t) =\mathbf{\tilde{X}}(t) + \mathbf{\tilde{W}}(t) =  \Gamma(t) \mathbf{A} + \mathbf{\tilde{W}}(t)
\end{align*}
where $\mathbf{A}=A^n=(A_1,\cdots,A_n)^{\mathrm{T}}$ and $\mathbf{\tilde{W}}(t) \sim \mathcal{N}({\bf 0}, \Gamma(t))$, and the corresponding Bayesian estimator will be
\begin{align*}
\hat{X}_t(Y^t) &= \E_{Q^*}[X_t|Y^t]\\
&=\sum_{i=1}^n \E_{Q^*}[A_i|\mathbf{\tilde{Y}}(t)]\phi_i(t).
\end{align*}
This implies that it is enough to find $\E_{Q^*} [A_i|\mathbf{\tilde{Y}}]$.

If $\Gamma(t)$ is invertible, this problem is simple. If $\Gamma(t)$ is not invertible, we can use the following tricks. Suppose the eigenvalue decomposition of matrix $\Gamma(t)$ is $\Gamma(t) = V(t)\Lambda(t) V(t)^{\mathrm{T}}$ where $V(t)=[v_1(t),\cdots,v_n(t)]$ is an orthonormal matrix and $\Lambda(t) = \diag(\lambda_1(t),\lambda_2(t),\cdots,\lambda_n(t))$ with $0\leq \lambda_1(t)\leq \lambda_2(t)\leq\cdots\leq \lambda_n(t)$. We can rewrite the problem as

\begin{align*}
V(t)^{\mathrm{T}} \mathbf{\tilde{Y}}(t) = \Lambda(t) V(t)^{\mathrm{T}} \mathbf{A} + V(t)^{\mathrm{T}}\mathbf{\tilde{W}}(t).
\end{align*}
Clearly we have $V(t)^T\mathbf{\tilde{W}}(t) \sim \mathcal{N}({\bf 0}, \Lambda(t))$. Let $m$ be the number of zero eigenvalues, i.e., $\lambda_1(t)=\cdots=\lambda_m(t) =0 <\lambda_{m+1}(t)$. As first $m$ elements can be removed, we can define effective vectors as
\begin{align*}
V_{\mbox{eff}}(t) &= [v_{m+1}(t) \cdots v_n(t)]\\
\Lambda_{\mbox{eff}}(t) &= \diag(\lambda_{m+1}(t), \cdots,\lambda_n(t)).
\end{align*}
Therefore, the above vector estimation problem can further be simplified as
\begin{align*}
V_{\mbox{eff}}(t)^{\mathrm{T}} \mathbf{\tilde{Y}}(t) = \Lambda_{\mbox{eff}}(t) V_{\mbox{eff}}(t)^{\mathrm{T}} \mathbf{A} + V_{\mbox{eff}}(t)^{\mathrm{T}}\mathbf{\tilde{W}}(t)\end{align*}
which is equivalent to
\begin{align}
&\Lambda_{\mbox{eff}}(t)^{-1/2}V_{\mbox{eff}}(t)^{\mathrm{T}} \mathbf{\tilde{Y}}(t) \nonumber\\
&= \Lambda_{\mbox{eff}}(t)^{1/2} V_{\mbox{eff}}(t)^{\mathrm{T}} \mathbf{A} + \Lambda_{\mbox{eff}}(t)^{-1/2}V_{\mbox{eff}}(t)^{\mathrm{T}}\mathbf{\tilde{W}}(t)\label{eq:equivalent_vector_estimation_problem}.
\end{align}
Note that $ \Lambda_{\mbox{eff}}(t)^{-1/2}V_{\mbox{eff}}(t)^{\mathrm{T}}\mathbf{\tilde{W}}(t)\sim \mathN(0,I_{n-m})$. Using equation \eqref{eq:equivalent_vector_estimation_problem}, we can easily find $\E[\mathbf{A}|Y^t] = \E[\mathbf{A}|\mathbf{\tilde{Y}}]$.

\subsubsection{Almost Optimal Causal Minimax Estimator}\label{subsubsec:almost optimum causal minimax estimator}
In Section \ref{subsubsec:bayesian estimator}, we show how to find $\E_{Q^*}[A|Y^t]$ if we know $Q^*$. However, it is often hard to find a capacity achieving distribution $Q^*$. Indeed most of the problems of finding capacity achieving distribution are still open including our sparse signal estimation problem $\sup_{P_{\theta}\in\mathP}I(A^n;B^n)$. Instead, we can use an approximated version of the prior $\tilde{Q}$. One natural choice of $\tilde{Q}$ is the capacity achieving distribution of $\sup_{P_{\theta}\in\mathP_{av}}I(A^n;B^n)$. This problem was recently considered by Zhang and Guo in \cite{zhang11}, where they referred to it as ``Gaussian channels with duty cycle and power constraints". They showed that the distribution on $A^n$ that maximizes this mutual information is i.i.d. and discrete. In other words, letting $P_d$ denote the supremum achieving distribution of
\begin{align*}
\sup_{P_A: \E[A^2]\leq P, P(A\neq 0)\leq q} I(A;B)
\end{align*}
 where $B=A+W$ and $W$ is a standard Gaussian noise $W$, then
\beas
\sup_{P_{\theta}\in\mathP_{av}}I(A^n;B^n) = n\left[ I(A;B)\right]_{P_A = P_d}
\eeas
where $\left[ I(A;B)\right]_{P_A = P_d}$ denotes the mutual information between $A$ and $B$ when the probability law of $A$ is $P_d$. Then, our choice of $\tilde{Q}$ will be $P_d^n$. The authors of \cite{zhang11} also showed that $P_d$ is discrete and has infinite number of mass points, and that it can be easily approximated with arbitrary precision. 

 Then the following question is the performance of this alternative filter compare to that of the minimax filter. More specifically, let define $L(\Lambda,\tilde{Q})$ by
\begin{align*}
L(\Lambda,\tilde{Q}) \stackrel{\triangle}{=} &\sup_{\theta\in\Lambda} R(\theta,\hat{X}_{\tilde{Q}}^{\rm Bayes}) - \min_{Q\in\mathP}\sup_{\theta\in\Lambda} \regretbayesQ.
\end{align*}
Following lemma gives an upper bound of $L(\Lambda,P_d^n)$.
\begin{lemma}\label{lem:upper bound of L}
\begin{align*}
L(\Lambda,P_d^n)\leq\left[I(A^n;B^n)\right]_{P_{A^n} = P_d^n} - \left[I(A^n;B^n)\right]_{P_{A^n} = Q^*}.
\end{align*}
\end{lemma}
Proof of Lemma \ref{lem:upper bound of L} is given in Appendix \ref{sec:proof of upper bound of L lemma}. This result implies that if these two mutual informations are close enough, then we are not losing much by using approximated version of optimum filter. Since $\left[I(A^n;B^n)\right]_{P_{A^n} = P_d^n}=n[I(A;B)]_{P_A=P_d}$, it is enough to argue that $n[I(A;B)]_{P_A=P_d}- \left[I(A^n;B^n)\right]_{P_{A^n} = Q^*}$ is small enough. The following lemma suggests that above two mutual informations are close for large $n$.
\begin{lemma}\label{lem:difference between two mutual information}
\begin{align*}
\lim_{n\rightarrow\infty}n\left[I(A;B)\right]_{P_A = P_d}- \sup_{P_{A^n}\in\mathP} I(A^n;B^n)=0
\end{align*}
\end{lemma}
Proof of Lemma \ref{lem:difference between two mutual information} is given in Appendix \ref{sec:proof of difference between two mutual information lemma}. Thus, if the number of basis are large enough, the performance of Bayesian filter $\hat{X}_{P_d^n}^{\rm Bayes}$ is close to the optimum.

\subsection{Poisson Channel and Direct Current Signal}\label{subsec:poisson channel and dc signal(examples)}
Consider direct current (DC) signal estimation over the Poisson channel. The input process is $X_t \equiv X$ for all $0\leq t\leq T$, where $X$ is a random variable bounded by $a\leq X\leq A$ for some positive constants $a$ and $A$. We can define the uncertainty set $\Lambda$ such that $\{P_{\theta} : \theta\in\Lambda\}$ is the set of all possible probability measures on $X$ under which $a\leq X\leq A$ holds almost surely. The estimator observes a Poisson process with rate $X_t$ and performance is measured under the \emph{natural loss function} $l(x,\hat{x}) = x\log(x/\hat{x}) -x+\hat{x}$.

Similar to the previous section, we can define $\Lambda_D$ and prove $\minimax(\Lambda)=\minimax(\Lambda_D)$. It is clear that $\{P_{\theta} : \theta\in\Lambda\}$ is convex and tight. Since $Y_T$ is a sufficient statistic of $Y^T$ for $X^T$ (which is constant at $X$), we have 
\begin{align*}
\minimax(\Lambda)&= \minimax(\Lambda_D)\\
&= \sup_{w\in\mu(\Lambda_D)}I(X^T;Y^T)\\
&=\sup_{P_X\in\mu([a,A])} I(X;Y_T),
\end{align*}
where the maximization is over all distributions on $X$ supported on $[a, A]$. The corresponding communication problem is the capacity achieving problem of the discrete-time Poisson channel. Discrete-time Poisson channel takes nonnegative, real valued $X$ as an input, and outputs a Poisson random variable with parameter $TX$. Note that we have additional input constraint that $a\leq X\leq A$ almost surely. In this scenario, Shamai \cite{shamai91} showed that capacity achieving distribution is discrete with finite number of mass points. Let $P_s$ be this capacity achieving distribution. Using Theorem \ref{thm:minimax filter}, we can conclude that the minimax causal estimator is conditional expectation of $X$ given $Y_t$ with respect to the distribution $P_s$, i.e.,
\beas
\hat{X}_t(Y^t) = \E_{P_s}[X|Y_t].
\eeas
Although an analytic expression of $P_s$ and capacity of the channel has yet to be found, we can approximate the distribution numerically to arbitrary precision.

\section{Experiments}\label{sec:experiments}
\subsection{Gaussian Channel and Sparse Signal}\label{subsec:gaussian channel and sparse signal(experiments)}
Consider the setting of Section \ref{subsec:gaussian channel and sparse signal(examples)}. As described in \cite{zhang11}, we approximate $P_d$ with finite number of mass points. Initially, find an maximum mutual information for three mass points, then increase the number of mass points until the increment of maximum mutual information is smaller than $10^{-5}$. Using approximated version of $P_d$, we can construct the Bayesian filter which is close to the optimum as described in Section \ref{subsubsec:almost optimum causal minimax estimator}.

 In order to compare the performance of the suggested minimax filter,  we introduce some possible estimators.
One naive choice of estimator is the maximum likelihood (ML) estimator. For equation \eqref{eq:equivalent_vector_estimation_problem}, the ML estimation of vector $\mathbf{A}$ is given as
\begin{align*}
\hat{\mathbf{A}} = \left(\Lambda_{\mbox{eff}}(t)^{1/2} V_{\mbox{eff}}(t)^{\mathrm{T}}\right)^{\dagger}\Lambda_{\mbox{eff}}(t)^{-1/2}V_{\mbox{eff}}(t)^{\mathrm{T}} \mathbf{\tilde{Y}}(t)
\end{align*}
where $X^{\dagger}$ is Moore-Penrose pseudo-inverse of matrix $X$. Since $A$ is sparse, we can further improve the estimator with thresholding. For example, estimator can do ML estimation and then take the largest $nq$ elements of $\hat{\mathbf{A}}$. 

Another possible estimator is the minimax estimator that lacks the sparsity information. Since the estimator does not know that the signal is sparse, it assumes the uncertainty set is $\mathP_{LS} = \{P_{\theta} : P_{\theta}(\frac{1}{n}||\mathbf{A}||_2^2 \leq P)=1\}$. Using similar ideas in the previous section, we can relate this minimax optimization problem to the channel coding problem on the Gaussian channel with average power constraint. Moreover, we can find the almost minimax filter which is Bayesian with i.i.d. Gaussian prior, i.e., $\mathbf{A}\sim\mathcal{N}({\bf 0},PI_n)$. Note that this filter turns out to be linear which is easy to implement. Using the result of the previous section, we have
\begin{align*}
&\Lambda_{\mbox{eff}}(t)^{-1/2}V_{\mbox{eff}}(t)^{\mathrm{T}} \mathbf{\tilde{Y}}(t)\nonumber\\
 &= \Lambda_{\mbox{eff}}(t)^{1/2} V_{\mbox{eff}}(t)^{\mathrm{T}} \mathbf{A}+ \Lambda_{\mbox{eff}}(t)^{-1/2}V_{\mbox{eff}}(t)^{\mathrm{T}}\mathbf{\tilde{W}}(t).
\end{align*}
Since every components are Gaussian, we can easily compute the conditional expectation. Recall, $\mathbf{A} \sim \mathcal{N}({\bf 0},PI_n)$, and $\Lambda_{\mbox{eff}}(t)^{-1/2} V_{\mbox{eff}}(t)^{\mathrm{T}} \tilde{Y}(t) \sim \mathcal{N}({\bf 0},P\Lambda_{\mbox{eff}}(t)+I_{n-m})$. Therefore, 
\begin{align*}
&\E[\mathbf{A} | \Lambda_{\mbox{eff}}(t)^{-1/2} V_{\mbox{eff}}(t)^{\mathrm{T}}\mathbf{\tilde{Y}}(t)] \nonumber\\
&=PV_{\mbox{eff}}(t) \left(P\Lambda_{\mbox{eff}}(t)+I_{n-m}\right)^{-1}  V_{\mbox{eff}}(t)^{\mathrm{T}} \mathbf{\tilde{Y}}(t).
\end{align*}

We can also consider the genie-aided scheme which allows additional information of the source. Suppose the decoder knows the position of nonzeros $i_1,\cdots,i_k$, i.e., the estimator knows the fact that $A_{i_1}, \cdots, A_{i_k}$ are nonzero and all others are zero. Clearly, this scheme should outperform all other schemes. Let $\mathbf{A}_{\mbox{nonzero}}$ be a $k$ dimensional vector that consists of nonzero elements of ${\bf A}$. Since the decoder has additional information, it is enough to estimate $\mathbf{A}_{\mbox{nonzero}}$. Using similar argument from the minimax estimator that lacks the sparsity information, we can show that the optimum minimax estimator is a Bayesian estimator with prior $\mathcal{N}\left(\mathbf{0},\frac{nP}{k}I_k\right)$. Recall equation \eqref{eq:equivalent_vector_estimation_problem} and let $U_{\mbox{eff}}$ be a matrix consisting of columns of $\Lambda_{\mbox{eff}}(t)^{1/2} V_{\mbox{eff}}(t)^{\mathrm{T}}$ which coincides with nonzero position of $A$. Then we can rewrite the equation \eqref{eq:equivalent_vector_estimation_problem} as
\begin{align*}
&\Lambda_{\mbox{eff}}(t)^{-1/2}V_{\mbox{eff}}(t)^{\mathrm{T}} \mathbf{\tilde{Y}}(t) \nonumber\\
&=  U_{\mbox{eff}}\mathbf{A}_{\mbox{nonzero}}+ \Lambda_{\mbox{eff}}(t)^{-1/2}V_{\mbox{eff}}(t)^{\mathrm{T}}\mathbf{\tilde{W}}(t).
\end{align*}
It is clear that $\Lambda_{\mbox{eff}}(t)^{-1/2} V_{\mbox{eff}}(t)^{\mathrm{T}} \mathbf{\tilde{Y}}(t) \sim \mathcal{N}({\bf 0},P\Lambda_{\mbox{eff}}(t)+I_{n-m})$. Therefore, 
\begin{align*}
&\E[\mathbf{A}_{\mbox{nonzero}}| \Lambda_{\mbox{eff}}(t)^{-1/2} V_{\mbox{eff}}(t)^{\mathrm{T}}\mathbf{ \tilde{Y}}(t)] \nonumber\\
&= \frac{nP}{k} U_{\mbox{eff}}^{\mathrm{T}} (U_{\mbox{eff}} U_{\mbox{eff}}^{\mathrm{T}}+I_{n-m})^{-1} \Lambda_{\mbox{eff}}(t)^{-1/2} V_{\mbox{eff}}(t)^{\mathrm{T}} \mathbf{\tilde{Y}}(t).
\end{align*}

\begin{figure}
\includegraphics[width=.5\textwidth]{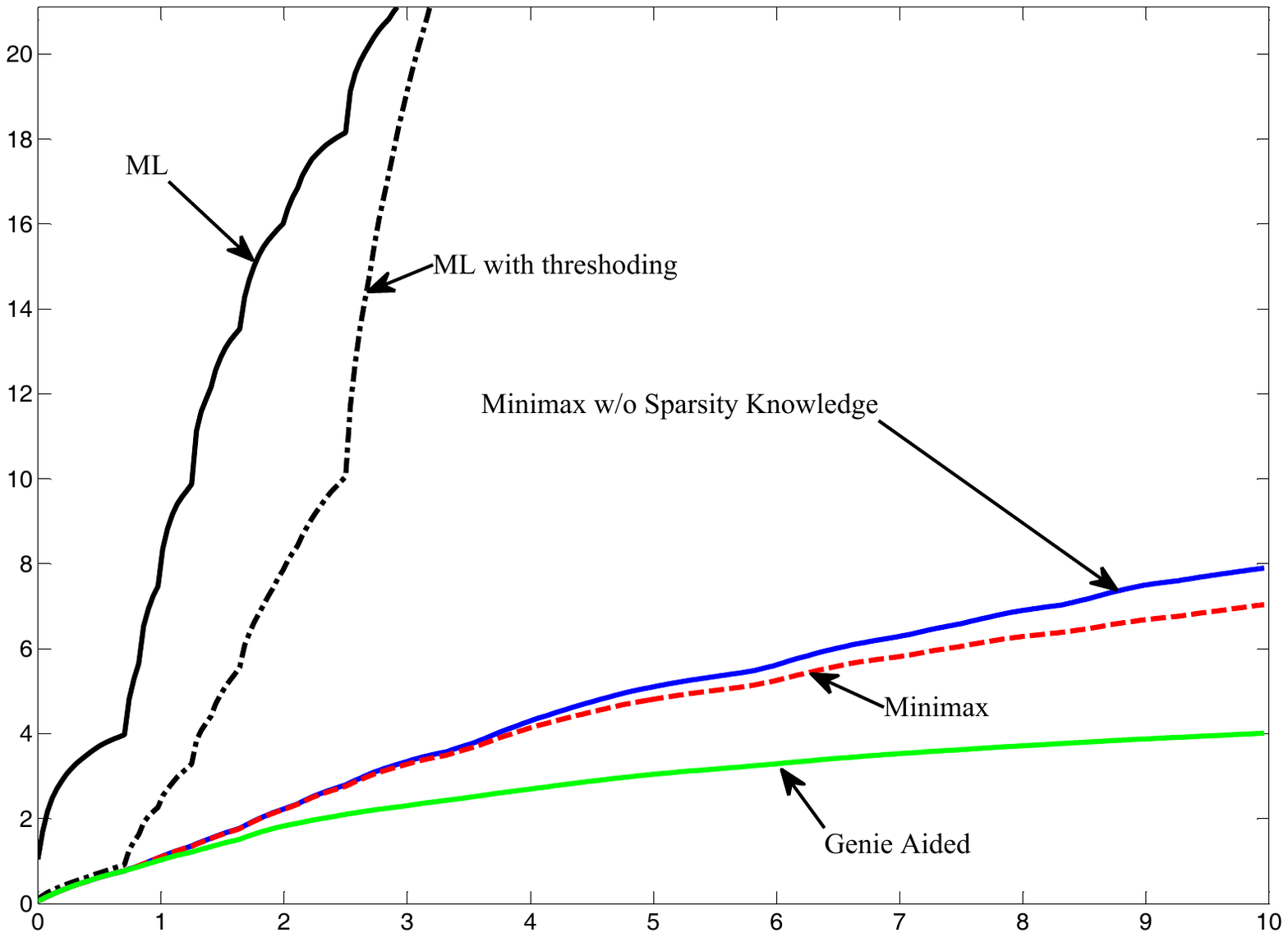}
\caption{Plots of $\cmle$ for the experiment of Section\ref{subsec:gaussian channel and sparse signal(experiments)}. We set $T=10$. $x$-axis shows time and the $y$-axis represents the worst causal mean loss for each estimator.}\label{fig:simulation results gaussian}
\end{figure}

We compare the performance of estimators in Figure \ref{fig:simulation results gaussian}. We choose $n=7$, $k=2$, $P = 10^{0.4}$(4dB), and Haar basis as an orthonormal signal set. We generate the random sparse coefficients by drawing the $k$ nonzero coefficients according to Gaussian distribution. For each realization of coefficients, we generate 100 output signals and take an average of causal loss. Finally, we take the maximum causal mean loss for each estimators among 100 simulations in order to check the worst case performance. We can see that minimax estimator outperforms maximum likelihood estimators and minimax estimator without sparsity knowledge. Note that the performance of minimax estimator is comparable to genie-aided estimator even though the genie-aided estimator used additional information.

\subsection{Poisson Channel and DC Signal}\label{subsec:poisson channel and dc signal(experiments)}
Optimum filter can be approximated using similar technique from Section \ref{subsec:gaussian channel and sparse signal(experiments)}. For comparison, we present some other natural estimators. First is the ML estimator,
\beas
\hat{X}_{ML}(Y^t) = \min \left\{ \max\{a,\frac{Y_t}{t}\}, A\right\}.
\eeas

Another possible estimator is a Bayesian estimator which assumes $X$ has uniform distribution, i.e., $X\sim U[a,A]$. In this case, the optimum Bayesian estimator is 
\begin{align*}
\hat{X}_{\mbox{unif}}(Y^t) = \frac{Y_t+1}{t} +\frac{e^{-at}a^{Y_t+1}-e^{-At}A^{Y_t+1}}{t\int_a^A e^{-xt}x^{Y_t} dx}.
\end{align*}
 
\begin{figure}
\centering
\includegraphics[width=.5\textwidth]{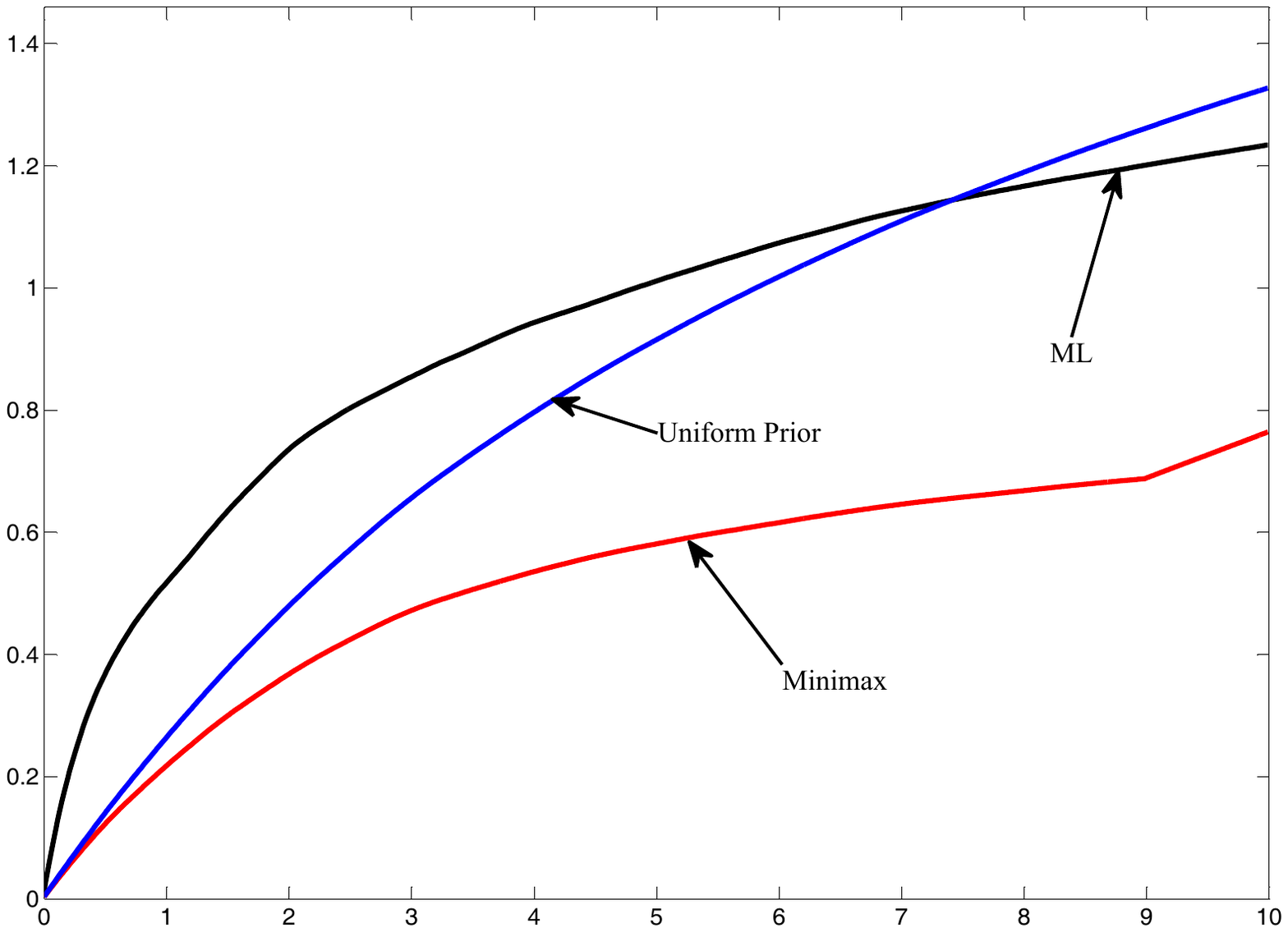}
\caption{Plots of $\cmle$ for the experiment of Section\ref{subsec:poisson channel and dc signal(experiments)}. Here we set $T=10$. $x$-axis shows time and the $y$-axis represents the worst causal mean loss for each estimator. }\label{fig:simulation results poisson}
\end{figure}

Figure \ref{fig:simulation results poisson} shows numerical results for $a=0.5$, $A=2$ case. We take an average of causal mean loss error over 100 times for $X=0.5,1,1.5,2$ and find an worst case error. The minimax estimator outperforms the other estimators as expected.

\section{Conclusions}\label{sec:conclusions}
We considered minimax estimation, focusing on the case of causal estimation when the noise-free object is a continuous-time signal and governed by a law from a given uncertainty set. We showed that the minimax filter is a Bayesian filter if the distortion criterion satisfies certain properties. We also characterized the worst case regret and the minimax estimator in the case of Gaussian and Poisson channels by relating it to a familiar communication problem of maximizing mutual information. We further showed that the capacity achieving prior coincides with the least favorable input. Using the idea of strong redundancy/regret-capacity theorem, we showed that our minimax estimator is optimum in a sense much stronger than it was designed to optimize for. Using these results, we presented two examples: sparse signal estimation under Gaussian setting and DC signal estimation under Poisson setting, for which we have used our results to derive and implement the minimax filter and exhibit its favorable performance in practice.

Our estimation framework can be applied to many other estimation problems. One possible extension is to apply Theorem \ref{thm:presence of feedback} to stochastic learning problems of the type considered by Bento et al. in \cite{bento11}. In this setting, the process $Y^T$ is defined by stochastic equation $Y_t = F(Y_t;A)dt+dW_t$, where $A$ is an unknown random parameter and $W^T$ is a standard Brownian motion. We can set $X_t = F(Y_t;A)$ and consider our estimation framework with feedback. We can apply our frameworks to estimate $X^T$ in the minimax sense and learn $A$. It will be interesting to investigate how an estimator guided by this approach would compare to that in \cite{bento11}.

\section*{Acknowledgment}\label{sec:Acknowledgment}
The authors would like to thank Ernest Ryu and Kartik Venkat for valuable discussions. The authors also would like to thank the anonymous reviewers and associate editor for their thorough and constructive feedback resulting in improved manuscript.

\appendices
\section{Proof of Lemma \ref{lem:theta_d is equal to theta}}\label{sec:proof of theta_d is equal to theta lemma}
Since $\Lambda_D\subset \Lambda$, we have
\begin{align*}
\minimax(\Lambda_D) &= \min_{Q\in\conv(\mathP_D)}\sup_{\theta\in\Lambda_D} \regretbayesQ\\
&=\min_{Q\in\mathP}\sup_{\theta\in\Lambda_D} \regretbayesQ\\
&\leq \min_{Q\in\mathP}\sup_{\theta\in\Lambda} \regretbayesQ\\
&=\minimax(\Lambda).
\end{align*}
On the other hand, 
\begin{align*}
&\minimax(\Lambda) \\
&= \min_{Q\in\mathP}\sup_{\theta\in\Lambda} \regretbayesQ\\
&\leq\min_{Q\in\mathP}\sup_{\theta\in\Lambda} \E_{P_{\theta}}\left[\int_0^T l(X_t,\E_Q[X|Y^t]) dt\right].
\end{align*}
It is clear that
\begin{align*}
&\E_{P_{\theta}}\left[\int_0^T l(X_t,\E_Q[X|Y^t]) dt\right]\\
&= \int\E\left[\int_0^T l(X_t,\E_Q[X|Y^t]) dt\suchthat A^n=a^n\right]dP_{\theta}(a^n),
\end{align*}
and therefore
\begin{align*}
&\sup_{\theta\in\Lambda} \E_{P_{\theta}}\left[\int_0^T l(X_t,\E_Q[X|Y^t]) dt\right]\\
&\leq\sup_{a^n\in \mathT^{(n)}} \E\left[\int_0^T l(X_t,\E_Q[X|Y^t]) dt\suchthat A^n=a^n\right]\\
&=\sup_{\theta\in\Lambda_D} \E_{P_\theta}\left[\int_0^T l(X_t,\E_Q[X|Y^t]) dt\right]
\end{align*}
where $\mathT^{(n)} = \{a^n\in\R^n : \frac{1}{n}\sum_{i=1}^n a_i^2 \leq P, \frac{1}{n}\sum_{i=1}^n \mathbf{1}(a_i\neq0) \leq q\}$ is a set of vector $a^n$ that satisfies constraints. This implies that 
\begin{align*}
\minimax(\Lambda)\leq&\min_{Q\in\mathP}\sup_{\theta\in\Lambda} \E_{P_{\theta}}\left[\int_0^T l(X_t,\E_Q[X|Y^t]) dt\right]\\
\leq&\min_{Q\in\mathP}\sup_{\theta\in\Lambda_D} \E_{P_\theta}\left[\int_0^T l(X_t,\E_Q[X|Y^t]) dt\right]\\
=& \min_{Q\in\mathP}\sup_{\theta\in\Lambda_D} \regretbayesQ\\
=&\minimax(\Lambda_D)
\end{align*}
%
Finally, these two inequalities imply 
\begin{align*}
\minimax(\Lambda) = \minimax(\Lambda_D).
\end{align*}
Indeed, 
\begin{align*}
\sup_{\theta\in\Lambda} \regretbayesQ = \sup_{\theta\in\Lambda_D} \regretbayesQ
\end{align*}
 holds for any $Q\in\mathP$ in general.


\section{Proof of Lemma \ref{lem:sufficient statistics}}\label{sec:proof of sufficient statistics lemma}

\begin{proof}
At time $t$, output process $Y^t$ can be discretized as
\begin{align*}
\bar{Y} =&\begin{bmatrix}Y_{\frac{t}{N}}\quad \left(Y_{\frac{2t}{N}}-Y_{\frac{t}{N}}\right)\quad\cdots\quad \left(Y_{\frac{Nt}{N}}-Y_{\frac{(N-1)t}{N}}\right)\end{bmatrix}^{\mathrm{T}}.
\end{align*}
This $\bar{Y}$ can be approximated as
\begin{align*}
\bar{Y} \approx \frac{1}{N}\bar{\Phi} A+\bar{W}
\end{align*}
where
\begin{align*}
\bar{\Phi} =&
\begin{bmatrix}
\phi_1(0) & \phi_2(0) & \cdots & \phi_n(0)\\\phi_1(\frac{t}{N}) & \phi_2(\frac{t}{N}) & \cdots & \phi_n(\frac{t}{N})\\ & & \vdots & \\\phi_1(\frac{(N-1)t}{N}) & \phi_2(\frac{(N-1)t}{N}) & \cdots &\phi_n(\frac{(N-1)t}{N})
\end{bmatrix}\\
A =& 
\begin{bmatrix}a_1\quad a_2 \quad\cdots\quad a_n\end{bmatrix}^{\mathrm{T}}\\
\bar{W} =&\begin{bmatrix}W_{\frac{t}{N}}\quad \left(W_{\frac{2t}{N}}-W_{\frac{t}{N}}\right)\quad\cdots\quad \left(W_{\frac{Nt}{N}}-W_{\frac{(N-1)t}{N}}\right)\end{bmatrix}^{\mathrm{T}}.
\end{align*}
It is easy to see that $\bar{W} \sim \mathcal{N}(0,\frac{1}{N} I_N)$. Furthermore, $\int_0^t \phi_i(s)dY_s$ can be approximated as
\begin{align*}
 \sum_{k=1}^N \phi_i\left(\frac{(k-1)t}{N}\right) \left(Y_{\frac{kt}{N}}-Y_{\frac{(k-1)t}{N}}\right).
\end{align*}
This approximation is similar to the idea from Ito's integral, and it is enough to prove the lemma based on this approximation. Therefore, the lemma holds if and only if $ p(A|\bar{Y}) = p(A|\bar{\Phi}^{\mathrm{T}} \bar{Y})$ for all $\bar{Y}$ which is enough to show that $\frac{p(\bar{Y}|A)}{p(\bar{\Phi}^{\mathrm{T}} \bar{Y}|A)}$ is constant (independent of choice of $A$) for all $\bar{Y}$. Throughout the proof, we assume $\bar{\Phi}^{\mathrm{T}}\bar{\Phi}$ is invertible, however, it is not difficult to derive the similar result when $\bar{\Phi}^{\mathrm{T}}\bar{\Phi}$ is not invertible.

It is easy to check that
\begin{align*}
&\log p(\bar{Y}|A)\\
 =& \log p(\bar{W} = \bar{Y}-\frac{1}{N}\bar{\Phi}A)\\
=&-\log {(2\pi (1/N)^N)^{N/2}} \\
&-\frac{N^N}{2}(\bar{Y}-\frac{1}{N}\bar{\Phi}A)^{\mathrm{T}}(\bar{Y}-\frac{1}{N}\bar{\Phi}A)\\
=&-\log {(2\pi (1/N)^N)^{N/2}}\\
&-\frac{N^N}{2}(\bar{Y}^{\mathrm{T}}\bar{Y}-\frac{2}{N}A^{\mathrm{T}}\bar{\Phi}^{\mathrm{T}}Y+\frac{1}{N^2}A^{\mathrm{T}}\bar{\Phi}^{\mathrm{T}}\bar{\Phi}A).
\end{align*}
On the other hand,
\begin{align*}
&\log p(\bar{\Phi}^{\mathrm{T}} \bar{Y}|A)\\
=&  \log p(\bar{\Phi}^{\mathrm{T}}\bar{W} = \bar{\Phi}^{\mathrm{T}}\bar{Y} - \frac{1}{N}\bar{\Phi}^{\mathrm{T}}\bar{\Phi}A)\\
=&-\log {(2\pi \cdot\mbox{det}((1/N)\bar{\Phi}^{\mathrm{T}}\bar{\Phi}))^{n/2}} \\
&-\frac{N^N}{2}(\bar{\Phi}^{\mathrm{T}}\bar{Y} - \frac{1}{N}\bar{\Phi}^{\mathrm{T}}\bar{\Phi}A)^{\mathrm{T}} (\bar{\Phi}^{\mathrm{T}}\bar{\Phi})^{-1}(\bar{\Phi}^{\mathrm{T}}\bar{Y} - \frac{1}{N}\bar{\Phi}^{\mathrm{T}}\bar{\Phi}A)\\
=&-\log {(2\pi \cdot\mbox{det}((1/N)\bar{\Phi}^{\mathrm{T}}\bar{\Phi}))^{n/2}}\\
&-\frac{N^N}{2}(\bar{Y}^{\mathrm{T}}\bar{Y} - \frac{2}{N}A^{\mathrm{T}}\bar{\Phi}^{\mathrm{T}}\bar{Y}+\frac{1}{N^2}A^{\mathrm{T}}\bar{\Phi}^{\mathrm{T}}\bar{\Phi}A)\\
&-\frac{N^N}{2}(\bar{Y}^{\mathrm{T}}\bar{\Phi} (\Phi^{\mathrm{T}}\Phi)^{-1}\bar{\Phi}^{\mathrm{T}}\bar{Y} - \bar{Y}^{\mathrm{T}}\bar{Y})
\end{align*}
where $\mbox{det}(\cdot)$ denotes the determinant of the matrix. Thus,
\begin{align*}
&\log \frac{p(\bar{Y}|A)}{p(\bar{\Phi}^{\mathrm{T}} \bar{Y}|A)}\\
=&\log  \frac{(2\pi\cdot \mbox{det}((1/N)\bar{\Phi}^{\mathrm{T}}\bar{\Phi}))^{n/2}}{(2\pi (1/N)^N)^{N/2}}\\
&+\frac{N^N}{2}(\bar{Y}^{\mathrm{T}}\bar{\Phi} (\Phi^{\mathrm{T}}\Phi)^{-1}\bar{\Phi}^{\mathrm{T}}\bar{Y} - \bar{Y}^{\mathrm{T}}\bar{Y}).
\end{align*}
Therefore, the fraction $\frac{p(\bar{Y}|A)}{p(\bar{\Phi}^{\mathrm{T}} \bar{Y}|A)}$ is independent of choice of $A$. This completes the proof of lemma. 
\end{proof}

\section{Proof of Lemma \ref{lem:upper bound of L}}\label{sec:proof of upper bound of L lemma}
\begin{proof}
Let define a class of all deterministic laws $\mathP_{D,all} = \{P_{\theta}: P_{\theta}(a^n)=1 \mbox{ for some $a^n\in\field{R}^n$}\}$ with corresponding index set $\Lambda_{D,all}$. Define $\mu_{D,av} = \{w\in\mu(\Lambda_{D,all}):\int P_{\theta}w(d\theta)\in\mathP_{av}\}$ which is a class of measure on $\Lambda_{D,all}$ that satisfies $\int P_{\theta}w(d\theta)\in\mathP_{av}$. Then,
\begin{align}
&\min_{Q\in\mathP_{av}} \sup_{w(\cdot)\in\mu_{D,av}} \int D(P_{\theta}||Q)w(d\theta)\nonumber\\
&= \min_{Q\in\mathP_{av}} \sup_{w(\cdot)\in\mu_{D,av}} \int D(P_{\theta}||Q_w)w(d\theta) + D(Q_w||Q)\label{eq:mu_av achiever}\\
&=  \sup_{w(\cdot)\in\mu_{D,av}}\min_{Q\in\mathP_{av}} \int D(P_{\theta}||Q_w)w(d\theta) + D(Q_w||Q)\label{eq:minimaxAgain}\\
&= \sup_{w(\cdot)\in\mu_{D,av}}\int D(P_{\theta}||Q_w)w(d\theta)\nonumber\\
&=  \sup_{w(\cdot)\in\mu_{D,av}}I(\Theta;B^n)\nonumber\\
&=  \sup_{w(\cdot)\in\mu_{D,av}}I(A^n;B^n)\nonumber\\
&=  \sup_{P_{A^n}\in\mathP_{av}}I(A^n;B^n)\nonumber\\
&=\left[I(A^n;B^n)\right]_{P_{A^n} = P_d^n}\nonumber
\end{align}
where we used minimax theorem in \eqref{eq:minimaxAgain}. Therefore, we can conclude that $P_d^n$ achieves the minimum of \eqref{eq:mu_av achiever}, i.e.,
\beas
 \sup_{w(\cdot)\in\mu_{D,av}} \int D(P_{\theta}||P_d^n)w(d\theta)=\left[I(A^n;B^n)\right]_{P_{A^n} =P_d^n}.
\eeas

On the other hand, we have
\begin{align*}
\sup_{\theta\in\Lambda} D(P_{\theta}||P_d^n) =&\sup_{\theta\in\Lambda_D} D(P_{\theta}||P_d^n) \\
=& \sup_{w(\cdot)\in\mu(\Lambda_D)} \int D(P_{\theta}||P_d^n) w(d\theta)\\
\leq& \sup_{w(\cdot)\in\mu_{D,av}} \int D(P_{\theta}||P_d^n)w(d\theta) \\
=&\left[I(A^n;B^n)\right]_{P_{A^n} = P_d^n}.
\end{align*}

Therefore, we can bound $L(\Lambda, P_d^n)$,
\begin{align*}
L(\Lambda,P_d^n) \stackrel{\triangle}{=} &\sup_{\theta\in\Lambda} R(\theta,\hat{X}_{P_d^n}^{\rm Bayes}) - \min_{Q\in\mathP}\sup_{\theta\in\Lambda} \regretbayesQ\\
\leq &\left[I(A^n;B^n)\right]_{P_{A^n} = P_d^n} - \left[I(A^n;B^n)\right]_{P_{A^n} = Q^*}.
\end{align*}
\end{proof}

\section{Proof of Lemma \ref{lem:difference between two mutual information}}\label{sec:proof of difference between two mutual information lemma}
\begin{proof}
It is trivial that $\sup_{w\in\mu(\Lambda)} I(A^n;B^n) \leq n\left[I(A;B)\right]_{P_A = P_d}$ for all $n$. Therefore, it is enough to find an upper bound of $ n\left[I(A;B)\right]_{P_A = P_d}-\sup_{w\in\mu(\Lambda)}I(A^n;B^n) $ that converges to 0 as $n$ grows. Recall that $ \sup_{w\in\mu(\Lambda)} I(A^n;B^n)$ is equal to $ \sup_{P_{\theta}\in\mathP} I(A^n;B^n)$.

Let probability law $P_{d,\epsilon}$ be a capacity achieving distribution of Gaussian channel with power constraint $P-\epsilon$ and duty cycle constraint $q-\epsilon$. In other words, $P_{d,\epsilon}$ is a supremum achiever of 
\beas 
\sup_{\substack{\E[A^2]\leq P-\epsilon\\ P(A\neq 0) <q-\epsilon}}I(A;B),
\eeas
 where $B$ is an output of standard Gaussian channel. Denote the measure $Q_p$ by projection of $P_{d,\epsilon}^n$ on $ \mathT^{(n)}_{\epsilon}$, i.e.,
\begin{align*}
Q_p(a^n) = 
\begin{cases}
\frac{P_{d,\epsilon}^n(a^n)}{\sum_{\tilde{a}^n\in \mathT^{(n)}_{\epsilon}}P_{d,\epsilon}^n(\tilde{a}^n)d\tilde{a}^n}&\mbox{if $a^n\in \mathT^{(n)}_{\epsilon}$}\\
0&\mbox{otherwise}
\end{cases}
\end{align*}
where $\mathT^{(n)}_{\epsilon}= \{a^n\in\R^n : P_{d,\epsilon}^n(a^n)\neq 0, \frac{1}{n}\sum_{i=1}^n a_i^2 \leq P, \frac{1}{n}\sum_{i=1}^n \mathbf{1}(a_i\neq0) \leq q\}$ is a set of point of masses $a^n$ that satisfies constraints. Alternatively, let $\mathN^{(n)}_{\epsilon}=\{a^n\in\R^n : P_{d,\epsilon}^n(a^n)\neq 0\}\setminus \left(\mathT^{(n)}_{\epsilon}\right)$, namely set of point masses that are not in the set $\mathT^{(n)}_{\epsilon}$. Recall that $P_{d,\epsilon}^n$ is discrete, and therefore both $Q_p$ and $P_{d,\epsilon}^n$ are probability mass functions. It is clear that $Q_p\in\mathP$ and $Q_p(a^n) = P_{d,\epsilon}^n(a^n|A^n\in \mathT_{\epsilon}^{(n)})$. Denote \begin{align*}
p^{(n)}_{\epsilon} \triangleq P_{d,\epsilon}^n(A^n\notin\mathT^{(n)}_{\epsilon})=P_{d,\epsilon}^n(A^n\in\mathN^{(n)}_{\epsilon}),
\end{align*}
then this implies 
\begin{align}
Q_p(a^n) = \frac{1}{1-p_{\epsilon}^{(n)}}P_{d,\epsilon}^n(a^n)\mathbf{1}(a^n\in \mathT_{\epsilon}^{(n)}).\label{eq:ratio between Qp and Pde}
\end{align}
By the law of large number, $p^{(n)}_{\epsilon}$ is vanishing exponentially as $n$ increase. Denote $Q_p(b^n)$ and $P_{d,\epsilon}^n(b^n)$ by output distributions of $B^n$ when the input law is $Q_p$ and $P_{d,\epsilon}^n$, respectively. Then, we have
\begin{align*}
&\left[I(A^n;B^n)\right]_{P_{A^n}=P_{d,\epsilon}^n} - \sup_{w\in\mu(\Lambda)} I(A^n;B^n) \nonumber\\
\leq&\left[I(A^n;B^n)\right]_{P_{A^n}=P_{d,\epsilon}^n} - \left[I(A^n;B^n)\right]_{P_{A^n}=Q_p}\\
=& \left[h(B^n)\right]_{P_{A^n}=P_{d,\epsilon}^n} -\left[h(B^n)\right]_{P_{A^n}=Q_p}\\
=& \int_{b^n} Q_p(b^n)\log Q_p(b^n)-P_{d,\epsilon}^n(b^n)\log P_{d,\epsilon}^n(b^n) db^n\\
=& D(Q_p(B^n)||P_{d,\epsilon}^n(B^n))\\
&+\int_{b^n} (Q_p(b^n)-P_{d,\epsilon}^n(b^n))\log P_{d,\epsilon}^n(b^n) db^n.
\end{align*}
Note that 
\begin{align}
Q_p(b^n) =& \sum_{a^n\in T_{\epsilon}^{(n)}} \frac{1}{1-p_{\epsilon}^{(n)}} P_{d,\epsilon}^n(a^n)P(b^n|a^n)\nonumber\\
\leq&\sum_{a^n} \frac{1}{1-p_{\epsilon}^{(n)}} P_{d,\epsilon}^n(a^n)P(b^n|a^n)\nonumber\\
=&\frac{1}{1-p_{\epsilon}^{(n)}} P_{d,\epsilon}^n(b^n),\label{eq:upperbound on Qpbn}
\end{align}
which implies
\begin{align*}
D(Q_p(B^n)||P_{d,\epsilon}^n(B^n)) \leq& -\log(1-p_{\epsilon}^{(n)}).
\end{align*}
By rearranging the terms, we can get
\begin{align*}
&-\int_{b^n} (P_{d,\epsilon}^n(b^n)-Q_p(b^n))\log P_{d,\epsilon}^n(b^n) db^n\\
=&-\int_{b^n} \left(\frac{1}{1-p_{\epsilon}^{(n)}}P_{d,\epsilon}^n(b^n)-Q_p(b^n)\right)\log P_{d,\epsilon}^n(b^n) db^n\\
&-\frac{p_e^{(n)}}{1-p_e^{(n)}} \cdot \left[h(B^n)\right]_{P_{A^n}=P_{d,\epsilon}^n}.
\end{align*}

We know $\frac{1}{1-p_{\epsilon}^{(n)}}P_{d,\epsilon}^n(b^n)-Q_p(b^n)$ is nonnegative for all $b^n$ from \eqref{eq:upperbound on Qpbn}. Also, we can bound $-\log P_{d,\epsilon}^n(b^n)$ using Jensen's inequality.
\begin{align*}
&-\log P_{d,\epsilon}^n(b^n) \nonumber\\
&= -\log\left(\sum_{a^n} P_{d,\epsilon}^n(a^n)P(b^n|a^n)da^n\right)\\
&\leq-\sum_{a^n}P_{d,\epsilon}^n(a^n) \log\left(\frac{1}{(\sqrt{2\pi})^n}\exp(-\frac{1}{2}||b^n-a^n||_2^2)\right)\\
&= n\log(\sqrt{2\pi}) +\frac{1}{2}\sum_{a^n} P_{d,\epsilon}^n(a^n)||b^n-a^n||_2^2 da^n\\
&\leq n\log(\sqrt{2\pi}) +||b^n||^2 + \E_{P^n_{d,\epsilon}}\left[||A^n||^2\right].
\end{align*}

Therefore,
\begin{align*}
&\left[I(A^n;B^n)\right]_{P_{A^n}=P_{d,\epsilon}^n} - \sup_{w\in\mu(\Lambda)} I(A^n;B^n) \nonumber\\
\leq&-\int_{b^n} \left(\frac{1}{1-p_{\epsilon}^{(n)}}P_{d,\epsilon}^n(b^n)-Q_p(b^n)\right)\log P_{d,\epsilon}^n(b^n) db^n\\
& +\delta_1^{(n)}\\
\leq&\frac{1}{1-p_{\epsilon}^{(n)}}(\E_{P_{d,\epsilon}^n}\left[||B^n||_2^2\right]+\E_{P_{d,\epsilon}^n}\left[||A^n||_2^2\right])\nonumber\\
&-(\E_{Q_p}\left[||B^n||_2^2\right]+\E_{P_{d,\epsilon}^n}\left[||A^n||_2^2\right])+\delta_1^{(n)}+\delta_2^{(n)}\\
=&\frac{1}{1-p_{\epsilon}^{(n)}}(2\E_{P_{d,\epsilon}^n}\left[||A^n||_2^2\right]+n)\nonumber\\
&-(\E_{Q_p}\left[||A^n||_2^2\right]+n+\E_{P_{d,\epsilon}^n}\left[||A^n||_2^2\right])+\delta_1^{(n)}+\delta_2^{(n)}\\
=& \frac{1+p_{\epsilon}^{(n)}}{1-p_{\epsilon}^{(n)}}\E_{P_{d,\epsilon}^n}\left[||A^n||_2^2\right]-\E_{Q_p}\left[||A^n||_2^2\right]+\delta_1^{(n)}+\delta_2^{(n)}\nonumber\\
&+ \frac{np_{\epsilon}^{(n)}}{1-p_{\epsilon}^{(n)}}\\
=& \sum_{a^n}\left(\frac{1}{1-p_{\epsilon}^{(n)}}P_{d,\epsilon}^n(a^n)-Q_p(a^n)\right)||a^n||_2^2 \nonumber\\
&+ \frac{p_{\epsilon}^{(n)}}{1-p_{\epsilon}^{(n)}}\E_{P_{d,\epsilon}^n}\left[||A^n||_2^2\right]+ \delta_1^{(n)}+\delta_2^{(n)}+\delta_3^{(n)}\\
=& \sum_{a^n\in \mathN_{\epsilon}^{(n)}}\left(\frac{1}{1-p_{\epsilon}^{(n)}}P_{d,\epsilon}^n(a^n)\right)||a^n||_2^2\nonumber\\
&+ \delta_1^{(n)}+\delta_2^{(n)}+\delta_3^{(n)}+\delta_4^{(n)}\\
=&\sum_{a^n}\left(\frac{1}{1-p_{\epsilon}^{(n)}}P_{d,\epsilon}^n(a^n)\right)||a^n||_2^2 \mathbf{1}(a^n\in\mathN_{\epsilon}^{(n)})da^n\nonumber\\
&+\delta_1^{(n)}+\delta_2^{(n)}+\delta_3^{(n)}+\delta_4^{(n)}\\
=&\frac{1}{1-p_{\epsilon}^{(n)}}\E_{P_{d,\epsilon}^n}\left[||A^n||_2^2 \mathbf{1}(A^n\in \mathN_{\epsilon}^{(n)})\right]\nonumber\\
&+\delta_1^{(n)}+\delta_2^{(n)}+\delta_3^{(n)}+\delta_4^{(n)},
\end{align*}
where $\delta_1^{(n)},\delta_2^{(n)},\delta_3^{(n)},\delta_4^{(n)}$ are defined as
\begin{align*}
\delta_1^{(n)} &=-\log(1-p_{\epsilon}^{(n)})-\frac{np_e^{(n)}}{1-p_e^{(n)}} \cdot \left[h(B)\right]_{P_{A}=P_{d,\epsilon}}\\
\delta_2^{(n)} &= \frac{p_{\epsilon}^{(n)}}{1-p_{\epsilon}^{(n)}}n\log(\sqrt{2\pi})\\
\delta_3^{(n)} &= \frac{np_{\epsilon}^{(n)}}{1-p_{\epsilon}^{(n)}} \\
\delta_4^{(n)} &= \frac{p_{\epsilon}^{(n)}}{1-p_{\epsilon}^{(n)}}\E_{P_{d,\epsilon}^n}[||A^n||_2^2],
\end{align*}
which are vanishing as $n$ grows to infinity. Note that $||A^n||_2^2 \mathbf{1}(A^n\in \mathN_{\epsilon}^{(n)})$ converges to zero with probability 1 by the strong law of large numbers, and therefore the expectation also converges to zero. By continuity of mutual information, we can finally conclude that $\left[I(A^n;B^n)\right]_{P_{A^n}=P_d^n} - \sup_{w\in\mu(\Lambda)} I(A^n;B^n)$ converges to zero as $n$ grows.
\end{proof}

\bibliographystyle{IEEEtran}

\begin{IEEEbiographynophoto}{Albert No} (S`12)
is currently a PhD candidate in the Department of Electrical Engineering at Stanford University, under the supervision of Prof. Tsachy Weissman.  His research interests include the relation between information and estimation theory, lossy compression and joint source-channel coding. 

Albert received a Bachelors degree in both Electrical Engineering and Mathematics from Seoul National University, in 2009, and a Masters degree in Electrical Engineering from Stanford University in 2012.
\end{IEEEbiographynophoto}

\begin{IEEEbiographynophoto}{Tsachy Weissman}(S'99-M'02-SM'07-F'12) 
is on the faculty of the department of Electrical Engineering at Stanford University, where he holds the STMicroelectronics Chair in the School of Engineering. He received his BSc and PhD from Technion in 1997 and 2001. He has published extensively on Information Theory, Statistical Signal Processing, the interplay between them, and their applications. He is inventor of several patents and involved in a number of hi-tech companies as researcher or member of the technical board. Much of his recent research has been dedicated to
the theory and practice of genomic data compression.

His research has been recognized with numerous awards, including best
paper awards, a Horev fellowship for Leaders in Science and
Technology, and a Henry Taub prize for excellence in research. He is a
Fellow of the Institute of Electrical and Electronics Engineers
(IEEE), and serves on the editorial boards of the IEEE Transactions on
Information Theory and Foundations and Trends in Communications and
Information Theory.
\end{IEEEbiographynophoto}
\end{document}